\documentclass[conference]{IEEEtran}
\IEEEoverridecommandlockouts

\usepackage{mathpazo}
\usepackage{times}
\usepackage{cite}
\usepackage{amsmath,amssymb,amsfonts}
\usepackage{algorithmic}
\usepackage{graphicx}
\usepackage{textcomp}
\usepackage{xcolor}
\usepackage{breqn}
\usepackage{mathtools}
\usepackage{bbm}
\usepackage{bm}
\usepackage{amsthm}
\usepackage{mathrsfs}
\usepackage{balance}

\DeclareSymbolFont{cmsymbols}{OMS}{cmsy}{m}{n}
\let\emptyset\relax
\DeclareMathSymbol{\emptyset}{\mathord}{cmsymbols}{"3B}

\def\BibTeX{{\rm B\kern-.05em{\sc i\kern-.025em b}\kern-.08em
    T\kern-.1667em\lower.7ex\hbox{E}\kern-.125emX}}

\newtheorem{lemma}{Lemma}
\newtheorem{theorem}{Theorem}

\newtheorem{remark}{Remark}

\begin{document}

\pagestyle{plain}

\title{Private Structured-Subset Retrieval}

\author{
\IEEEauthorblockN{Maha Issa and Anoosheh Heidarzadeh}
\IEEEauthorblockA{
Department of Electrical and Computer Engineering\\
Santa Clara University, Santa Clara, CA, USA\\
\{missa,aheidarzadeh\}@scu.edu
}
}

\maketitle

\thispagestyle{plain}

\begin{abstract}
We introduce the \emph{Private Structured-Subset Retrieval (PSSR)} problem, where a user retrieves $D$ messages from a database of $K$ messages replicated across $N$ non-colluding servers, and the demand is restricted to a known structured family of $D$-subsets. This formulation generalizes Multi-message Private Information Retrieval (MPIR) and captures settings where the demand space is constrained by application-specific structure. Focusing on balanced ${\{0,1\}}$-linear schemes, a class that includes several best-known MPIR schemes, we derive converse bounds on the maximum retrieval rate and minimum subpacketization level required to achieve any given rate. We also develop an optimization-based framework to construct schemes for general structured demand families, providing flexibility in optimizing the retrieval rate or the subpacketization level. When specialized to the full demand family, this framework recovers known balanced $\{0,1\}$-linear MPIR constructions; for more restricted demand families, it can exploit the demand structure to increase the retrieval rate, reduce the subpacketization level, or both. We demonstrate this through a structured-demand example in which the proposed PSSR scheme simultaneously achieves a higher rate and requires a smaller subpacketization than the best-known MPIR scheme for the same parameters $N$, $K$, and $D$. Our parallel work on contiguous-demand families further illustrates the scope of this framework by yielding rate-optimal schemes with substantially smaller subpacketization and no field-size restrictions, improving upon MPIR-based schemes.
\end{abstract}

\section{Introduction}

Private Information Retrieval (PIR) is a fundamental problem in information-theoretic privacy~\cite{SJ2017,SJ2016ArbitraryTIFS,SJ2018Multiround,TSC2019,VBU2022}. 
In this problem, a user wishes to retrieve a message from a database stored across one or multiple remote servers, without revealing the identity of the desired message. 
Although downloading the entire database trivially preserves privacy, it incurs a prohibitive download cost when the database contains many messages. 
This motivates the design of PIR schemes that maximize the retrieval rate, defined as the ratio of the number of desired message symbols to the total number of retrieved symbols.

Several variations of the PIR problem have been studied in the literature
(see, e.g.,~\cite{VWU2023,UAGJTT2022}, and references therein). 
Among these variants, the Multi-message PIR (MPIR) problem, originally introduced in~\cite{BU2018} and further studied in~\cite{WHS2022,HWS2025,WHS2025}, generalizes classical PIR by allowing the user to retrieve multiple messages in a single protocol. 
Specifically, the user wishes to retrieve ${D \geq 2}$ messages from a database of $K$ messages. 
In MPIR, the query received by any server must not reveal which of the $\binom{K}{D}$ possible $D$-subsets is requested. 
Capacity-achieving MPIR schemes are known for several regimes, including ${D \geq K/2}$ and ${D \mid K}$~\cite{BU2018}, while improved constructions for certain remaining regimes were developed in~\cite{WHS2022,HWS2025,WHS2025}. 
These schemes extend capacity-achieving PIR schemes for the single-message case~\cite{SJ2017,TSC2019}.

In many practical applications, however, it may be known to the servers that the user's demand belongs to a structured family of $D$-subsets rather than to the full collection of all $\binom{K}{D}$ subsets. 
For example, in a medical dataset containing demographic and clinical information, a research entity may be interested only in retrieving records corresponding to patients who share a common condition, such as diabetes or hypertension. 
Similarly, in applications involving ordered or segmented data, a user may be interested only in consecutive time periods or data segments, rather than arbitrary subsets. 
Thus, in such scenarios, the demand index sets are precisely the $D$-subsets satisfying the structure imposed by the application. 
Consequently, privacy only requires that, from the query received by any server, the server cannot identify which subset within this candidate family is the user's demand index set.

Existing MPIR schemes can be directly applied to this setting, since they guarantee privacy over the full collection of $\binom{K}{D}$ subsets and therefore also over any restricted family. 
However, applying MPIR schemes to such structured settings may be suboptimal in terms of the retrieval rate, since MPIR schemes are designed for the full demand space rather than for a structured demand family. 
Moreover, many existing MPIR schemes are linear and require messages to be divided into subpackets~\cite{BU2018,HWS2025}. 
A large subpacketization level can increase overhead and implementation complexity. 
By exploiting the structure of the demand family, it may be possible to improve the retrieval rate, reduce subpacketization, and obtain simpler, more efficient schemes. 
This motivates the development of retrieval schemes tailored specifically to structured-demand settings.

In this paper, we introduce the \emph{Private Structured-Subset Retrieval (PSSR)} problem. 
In this problem, a database of $K$ messages is replicated across $N$ non-colluding servers, and a user wishes to retrieve $D$ messages whose indices are drawn from a known family of candidate $D$-subsets of ${[1:K]}$. 
The user aims to prevent each server from identifying which subset in this family is requested, while maximizing the retrieval rate.

We focus on a class of schemes termed \emph{balanced ${\{0,1\}}$-linear PSSR schemes}. 
In such schemes, each message is partitioned into $L$ subpackets; 
the user queries each server using ${\{0,1\}}$-linear combinations of message subpackets; and each server returns the corresponding linear combinations as their answer.
Because the queries/answers involve only additions of subpackets, these schemes are simple to implement and work over any field. 
The balanced structure further ensures the user forms and sends equal-size queries to all servers, and each server computes and returns an equal-size answer. 
Moreover, several best-known schemes for PIR and MPIR, which are special cases of PSSR, belong to this class, including the PIR scheme of~\cite{SJ2017} and the MPIR schemes of~\cite{BU2018,HWS2025}. 
Thus, balanced ${\{0,1\}}$-linear schemes provide a natural and rich framework for designing efficient PSSR schemes and analyzing the rate--subpacketization tradeoff.

Our goal is to characterize the fundamental limits of the PSSR problem within this class of schemes. 
To this end, we derive converse bounds on the maximum achievable retrieval rate and the minimum required subpacketization level. 
We also develop an optimization-based framework for constructing PSSR schemes that achieve, or closely approach, these bounds. 
This framework extends the approach of~\cite{HWS2025}, originally developed for MPIR, to general structured demand families by adapting the query and answer patterns to the given structure rather than treating all messages uniformly as in~\cite{HWS2025}.

Our results show that the PSSR rate upper bound can strictly exceed the best-known MPIR rate upper bound for the same values of $N$, $K$, and $D$.
This indicates that rate-optimal MPIR schemes may be suboptimal for PSSR, in which case our proposed framework can exploit the restricted demand structure to achieve better rates. 
Additionally, by treating the subpacketization level as an optimization variable rather than a byproduct of the construction, the asymmetry of the structured family can be leveraged to obtain schemes with smaller subpacketization than MPIR-based solutions, even when the optimal retrieval rate is unchanged.
To illustrate these findings, an example with a restricted demand family is presented in which the proposed PSSR scheme achieves the maximum achievable retrieval rate with the minimum required subpacketization level among balanced linear schemes.
For the same values of $N$, $K$, and $D$, the best-known MPIR scheme achieves a strictly lower retrieval rate and requires a larger subpacketization level, showing that restricting the demand family can improve retrieval efficiency and reduce implementation complexity.

Finally, although the tightness of the general PSSR rate and subpacketization bounds remains open, our parallel work~\cite{IH2026PCBRarXiv} shows that, when the proposed framework is specialized to the demand family of all ${K-D+1}$ contiguous length-${D}$ message blocks, it yields rate-optimal schemes for all $K$ and $D$, with subpacketization levels optimal for a broad range of parameter values. 
This specialization improves upon existing MPIR-based schemes in different regimes: it achieves the optimal rate when known MPIR schemes are not rate-optimal, reduces subpacketization when MPIR schemes are rate-optimal, and works over any field when MPIR schemes with the same rate and subpacketization require sufficiently large field sizes.

\section{Problem Setup}

For integers $i$ and $j$ with ${0\leq i\leq j}$, we denote the set ${\{i,i+1,\dots,j\}}$ by ${[i:j]}$.
We denote random variables by bold-face symbols and their realizations by regular symbols.
Throughout, we fix an arbitrary prime power $q$, denote the finite field of order $q$ by $\mathbb{F}_q$, and denote the $L$-dimensional vector space over $\mathbb{F}_q$ by $\mathbb{F}_q^L$ for any integer $L\geq 1$.
We also denote the set of positive integers by $\mathbb{N}$ and the set of nonnegative integers by $\mathbb{N}_0$.

Consider a dataset consisting of $K$ messages ${\mathrm{X}_1,\dots,\mathrm{X}_K}$, replicated across ${N\geq 2}$ non-colluding servers. 
Suppose each message $\mathrm{X}_i$ for ${i \in [1:K]}$ consists of $L$ symbols from ${\mathbb{F}_q}$. 
That is, ${\mathrm{X}_i\in \mathbb{F}_{q}^{L}}$. 
For any non-empty $\mathrm{U} \subseteq [1:K]$, let ${\mathrm{X}_{\mathrm{U}} \coloneqq \{\mathrm{X}_i: i\in \mathrm{U}\}}$ denote the set of messages indexed by $\mathrm{U}$.

Suppose a user wishes to retrieve $D$ messages, for some ${D\in[2:K-1]}$, indexed by $\mathrm{W}\in \{\mathrm{W}_1,\dots,\mathrm{W}_E\}$, where each $\mathrm{W}_j$ is a subset of $[1:K]$ of size $D$.

We refer to $\mathrm{X}_{\mathrm{W}}$ as the \emph{demand messages},  ${\mathrm{X}_{[1:K]\setminus \mathrm{W}}}$ as the \emph{interference messages}, $\mathrm{W}$ as the \emph{demand index set}, and ${\mathrm{W}_1,\dots,\mathrm{W}_E}$ as the \emph{candidate demand index sets}. 

We assume that the random variables 
${\mathbf{X}_1,\dots,\mathbf{X}_K}$ are independent and uniformly distributed over ${\mathbb{F}_{q}^{L}}$; 
the random variable $\mathbf{W}$ is distributed arbitrarily over $\{\mathrm{W}_1,\dots,\mathrm{W}_E\}$, subject to the condition that every $\mathrm{W}_j$ has a nonzero probability; and the random variables ${\mathbf{X}_{[1:K]}}$ and $\mathbf{W}$ are independent.

The user generates $N$ queries ${\mathrm{Q}^{[\mathrm{W}]}_{1},\dots,\mathrm{Q}^{[\mathrm{W}]}_{N}}$, and sends query ${\mathrm{Q}^{[\mathrm{W}]}
_{n}}$ to server $n$ for each $n\in [1:N]$. 
Each query is a (possibly stochastic) function of the demand index set, constructed without a prior access to the messages, i.e., 
\begin{equation}
\label{eq:Q_indep_X_S}
I(\mathbf{Q}_{[1:N]}^{[\mathrm{W}]};\mathbf{X}_{[1:K]})=0,
\end{equation} where $\mathbf{Q}_{[1:N]}^{[\mathrm{W}]}\coloneqq \{\mathbf{Q}^{[\mathrm{W}]}_{1},\dots,\mathbf{Q}^{[\mathrm{W}]}_{N}\}$. 

Upon receiving the query ${\mathrm{Q}_{n}^{[\mathrm{W}]}}$, each server $n$ computes an answer ${\mathrm{A}_{n}^{[\mathrm{W}]}}$ and sends it back to the user. 
Answers are deterministic functions of the queries and messages. 
That is, 
\begin{equation} 
\label{eq:A_func_Q_X_S}
H(\mathbf{A}_{n}^{[\mathrm{W}]} | \mathbf{Q}_{n}^{[\mathrm{W}]}, \mathbf{X}_{[1:K]})=0, \quad \forall n \in [1:N].
\end{equation}

Once all the servers' answers are received, the user must be able to recover the demand messages, i.e.,
\begin{equation} 
\label{eq:correctness}
H( \mathbf{X}_{\mathrm{W}} | \mathbf{Q}_{[1:N]}^{[\mathrm{W}]}, \mathbf{A}_{[1:N]}^{[\mathrm{W}]})=0,
\end{equation} where $\mathbf{A}_{[1:N]}^{[\mathrm{W}]}\coloneqq \{\mathbf{A}^{[\mathrm{W}]}_{1},\dots,\mathbf{A}^{[\mathrm{W}]}_{N}\}$. 
We refer to this requirement as the \emph{correctness condition}. 

The information available to any server must reveal no information regarding the realization $\mathrm{W}$. 
That is,
\begin{equation}
\label{eq:privacy}	I(\mathbf{W};\mathbf{Q}_{n}^{[\mathrm{W}]},\mathbf{A}_{n}^{[\mathrm{W}]},\mathbf{X}_{[1:K]})=0, \quad \forall n \in [1:N].
\end{equation}
This requirement, which we refer to as the \emph{privacy condition}, keeps the user's demand index set private from any server.

The problem is to design a scheme that simultaneously satisfies the correctness and privacy conditions. 
We refer to this problem as \emph{Private Structured-Subset Retrieval (PSSR)}. 

Without loss of generality, we restrict attention to instances satisfying $\big|\bigcup_{j=1}^{E}\mathrm{W}_j\big|=K$ and $\big|\bigcap_{j=1}^{E}\mathrm{W}_j\big|=0$.
Indeed, if some index ${i\in[1:K]}$ belongs to none of the sets $\mathrm{W}_j$, then the message $\mathrm{X}_i$ is irrelevant and can be removed, yielding an equivalent instance with ${K-1}$ messages. 
Likewise, if some index $i$ belongs to all of the sets $\mathrm{W}_j$, then $\mathrm{X}_i$ can be recovered separately without affecting correctness or privacy, and removing this index from every $\mathrm{W}_j$ again yields an equivalent instance with ${K-1}$ messages.

When ${D\mid K}$ and ${E = K/D}$, the PSSR problem is equivalent to the classical single-message PIR problem~\cite{SJ2017}: each candidate demand subset can be viewed as a super-message consisting of $K/E$ messages, and the user privately retrieves one such super-message from the $E$ super-messages stored across the $N$ servers.
Additionally, when ${E=\binom{K}{D}}$, the PSSR problem reduces to the MPIR problem~\cite{BU2018}. 

In this work, we focus on a class of PSSR schemes that we call \emph{balanced $\{0,1\}$-linear PSSR schemes}. 
In such schemes, each message is partitioned into $L$ subpackets, each consisting of a single message symbol. The user queries each server for a collection of linear combinations of message subpackets with coefficients in $\{0,1\}$, and
each server answers the user with the corresponding linear combinations.
Moreover, the queries sent to all servers have the same total length, and the answers returned by all servers have the same total length.

We define the \emph{retrieval rate} of a scheme as the ratio between the amount of information required by the user and the total amount of information retrieved from all servers, namely
\begin{equation}\label{eq:RateDef}
\frac{H(\mathbf{X}_{\mathbf{W}})}{\sum_{n=1}^N H (\mathbf{A}_{n}^{[\mathbf{W}]} | \mathbf{Q}_{n}^{[\mathbf{W}]})}.
\end{equation}
We also define the \emph{subpacketization level} as the number of subpackets into which the scheme partitions each message.

The goals of this work are fourfold:
\begin{itemize}
\item[(i)] to establish an upper bound $R^{*}$ on the maximum retrieval rate achievable by balanced ${\{0,1\}}$-linear PSSR schemes over all subpacketization levels, in terms of the number of servers $N$ and the candidate demand index sets $\mathrm{W}_1,\dots,\mathrm{W}_E$ (and, in turn, the total number of messages $K$ and the number of demand messages $D$);
\item[(ii)] to construct a scheme that achieves a rate $R_{*}$ matching (or closely approaching) the bound $R^{*}$;
\item[(iii)] to establish a lower bound $L_{*}$ on the minimum subpacketization level required to achieve rate $R_{*}$; and
\item[(iv)] to construct a scheme achieving rate $R_{*}$ with a subpacketization level $L^{*}$ matching (or close to) the bound $L_{*}$.
\end{itemize}

\section{Main results}

This section presents our main converse and achievability results for the PSSR problem.   

To simplify notation, for each ${j\in[1:E]}$, each ${\mathrm{S}\subseteq[1:K]}$, and each ${l\in[1:D]}$, define
\[
\mathbbm{V}_{j,l}(\mathrm{S})
\coloneqq
\bigl\{
\mathrm{V}\subseteq \mathrm{W}_j\setminus \mathrm{S}:\ |\mathrm{V}|=l
\bigr\},
\]
and let $\mathbbm{V}_{j}(\mathrm{S})\coloneqq \bigsqcup_{l=1}^{D}\mathbbm{V}_{j,l}(\mathrm{S})$. 
Similarly, for each ${l\in[0:D-1]}$, define
\[
\mathbbm{U}_{j,l}(\mathrm{S})
\coloneqq
\bigl\{
\mathrm{U}\subseteq [1:K]\setminus \mathrm{S}:\ 
\mathrm{U}\not\subseteq \mathrm{W}_j,\ 
|\mathrm{U}\cap \mathrm{W}_j|=l
\bigr\},
\]
and let $\mathbbm{U}_{j}(\mathrm{S})\coloneqq \bigsqcup_{l=0}^{D-1}\mathbbm{U}_{j,l}(\mathrm{S})$. 

\begin{theorem}
\label{thm:PRSR capacity}
For $N$ servers and $E$ candidate demand index sets $\mathrm{W}_1,\dots,\mathrm{W}_E$, 
the maximum retrieval rate achievable by any balanced $\{0,1\}$-linear PSSR scheme is upper bounded by
\begin{align}
\label{eq:PSR ub}
R^{*} &\coloneqq D  
  \left( \max 
 \sum_{j=1}^{E} \frac{1}{N^{j-1}} \left| {\mathrm{W}_{\pi(j)}} \setminus \textstyle\bigcup_{k=0}^{j-1} {\mathrm{W}_{\pi(k)}} \right| \right)^{-1},
\end{align}
where $\pi(0)\coloneqq 0$,  $\mathrm{W}_{0}\coloneqq \emptyset$, and the maximization is over all permutations $\pi: [1:E]\rightarrow [1:E]$, 
and is lower bounded by 
\begin{equation}
\label{eq:PSSR lb}
R_{*} \coloneqq \frac{D}{N} \left(\min \frac{1}{L}\sum_{\substack{\mathrm{U}\subseteq [1:K]}} T_{\mathrm{U}} \right)^{-1},
\end{equation}
where the minimization is over the integer variables $\{T_{\mathrm{U}}\}$, $\{I_{\mathrm{U},\mathrm{V}}^{[\mathrm{W}_j]}\}$, $\{J^{[\mathrm{W}_j]}_{\mathrm{V},i}(k)\}$, and $L$, subject to the constraints~\eqref{eq:set_1_cons}--\eqref{eq:L_geq_1}.
Here, $T_{\mathrm{U}}$ denotes the number of retrieved linear combinations per server with message support $\mathrm{U}\subseteq[1:K]$. 
The auxiliary variables $I_{\mathrm{U},\mathrm{V}}^{[\mathrm{W}_j]}$ and $J_{\mathrm{V},i}^{[\mathrm{W}_j]}(k)$, formally defined in Section~\ref{sec:Ach_Proofs}, respectively track the interference-cancellation steps and the subsequent recovery of demand-message subpackets. 
Finally, $L$ denotes the subpacketization level of the scheme.

\begin{figure}[t]
\centering
\begin{minipage}{\columnwidth}\vspace{-0.04cm}
\begin{align}
& \sum_{\mathrm{V}\in \mathbbm{V}_{j}(\mathrm{U})}
I^{[\mathrm{W}_{j}]}_{\mathrm{U},\mathrm{V}}
+
(N-1)
\sum_{\mathrm{V}\in \mathbbm{V}_j([1:K]\setminus \mathrm{U})}
I^{[\mathrm{W}_{j}]}_{\mathrm{U}\setminus \mathrm{V},\mathrm{V}}
\leq T_{\mathrm{U}},
\nonumber\\
& \hspace{1.5em}
\forall j \in [1:E],\ 
\forall \mathrm{U}\subseteq[1:K],\
\mathrm{U}\not\subseteq \mathrm{W}_j ,
\label{eq:set_1_cons}
\\[0.4em]
& T_{\{i\}}
+
(N-1)
\sum_{\mathrm{U}\in \mathbbm{U}_j(\{i\})}
I^{[\mathrm{W}_{j}]}_{\mathrm{U},\{i\}}
\nonumber\\
& \quad
+
\sum_{k=2}^{D}
\sum_{l=1}^{k-1}
\sum_{\mathrm{V}\in \mathbbm{V}_{j,l}(\{i\})}
J^{[\mathrm{W}_{j}]}_{\mathrm{V}\cup\{i\},i}(k)
=
\frac{L}{N},
\nonumber\\
& \hspace{1.5em}
\forall j\in[1:E],\
\forall i\in \mathrm{W}_{j},
\label{eq:set_2_cons}
\\[0.4em]
& T_{\mathrm{V}}
+
(N-1)
\sum_{\mathrm{U}\in \mathbbm{U}_j(\mathrm{V})}
I^{[\mathrm{W}_{j}]}_{\mathrm{U},\mathrm{V}}
\geq
\sum_{k=|\mathrm{V}|}^{D}
\sum_{i\in \mathrm{V}}
J^{[\mathrm{W}_{j}]}_{\mathrm{V},i}(k),
\nonumber\\
& \hspace{1.5em}
\forall j\in[1:E],\
\forall \mathrm{V}\subseteq \mathrm{W}_j,\
|\mathrm{V}|\geq 2,
\label{eq:set_3_cons}
\\[0.4em]
& (N-1)T_{\{i\}}
+
(N-1)^2
\sum_{k=1}^{m}
\sum_{\mathrm{U}\in \mathbbm{U}_{j,k-1}(\{i\})}
I^{[\mathrm{W}_{j}]}_{\mathrm{U},\{i\}}
\nonumber\\
& \quad
+
(N-1)
\sum_{k=2}^{m}
\sum_{l=1}^{k-1}
\sum_{\mathrm{V}\in \mathbbm{V}_{j,l}(\{i\})}
J^{[\mathrm{W}_{j}]}_{\mathrm{V}\cup\{i\},i}(k)
\nonumber\\
& \geq
N
\sum_{k=2}^{m+1}
\sum_{l=1}^{k-1}
\sum_{\mathrm{V}\in \mathbbm{V}_{j,l}(\{i\})}
\sum_{\mathrm{U}\in
\mathbbm{U}_{j,k-1-l}(\mathrm{V}\cup\{i\})}
I^{[\mathrm{W}_{j}]}_{\mathrm{U}\cup\{i\},\mathrm{V}}
\nonumber\\
& \quad
+
\sum_{k=2}^{m+1}
\sum_{l=1}^{k-1}
\sum_{\mathrm{V}\in \mathbbm{V}_{j,l}(\{i\})}
\sum_{i'\in \mathrm{V}}
J^{[\mathrm{W}_{j}]}_{\mathrm{V}\cup\{i\},i'}(k),
\nonumber\\
& \hspace{1.5em}
\forall j\in[1:E],\
\forall i\in \mathrm{W}_{j},\
\forall m\in[1:D-1],
\label{eq:indexing_cons_1}
\\[0.4em]
& \sum_{\substack{\mathrm{U}\subseteq[1:K]\\ i\in \mathrm{U}}}
T_{\mathrm{U}}
\leq L,
\quad
\forall i\in[1:K],
\label{eq:indexing_cons_2}
\\[0.4em]
& T_{\mathrm{U}}\in\mathbb{N}_0,
\quad
\forall \mathrm{U}\subseteq[1:K],\
|\mathrm{U}|\geq 1,
\label{eq:T_non_neg}
\\[0.4em]
& I^{[\mathrm{W}_{j}]}_{\mathrm{U},\mathrm{V}}
\in\mathbb{N}_0,
\quad
\forall j\in[1:E],\
\forall \mathrm{V}\subseteq \mathrm{W}_j,\
|\mathrm{V}|\geq 1,
\nonumber\\
& \hspace{1.5em}
\forall \mathrm{U}\subseteq[1:K]\setminus\mathrm{V},\
\mathrm{U}\not\subseteq \mathrm{W}_j,
\label{eq:I_non_neg}
\\[0.4em]
& J^{[\mathrm{W}_{j}]}_{\mathrm{V},i}(k)
\in\mathbb{N}_0,
\quad
\forall j\in[1:E],\
\forall \mathrm{V}\subseteq \mathrm{W}_j,\
|\mathrm{V}|\geq 2,
\nonumber\\
& \hspace{1.5em}
\forall i\in\mathrm{V},\
\forall k\in[|\mathrm{V}|:D],
\label{eq:J_non_neg}
\\[0.4em]
& L\in\mathbb{N}.
\label{eq:L_geq_1}
\end{align}
\end{minipage}\vspace{-0.6cm}
\end{figure}
\end{theorem}

The upper bound holds for all PSSR schemes (and hence for the balanced $\{0,1\}$-linear schemes considered here), and its proof follows the same steps as in~\cite[Theorem~1]{CWJ2018}; although~\cite{CWJ2018} adopts a different model for message dependence (see~\cite{CWJ2018} for details), the argument applies directly to arbitrary dependence models, including the one considered here. 

To prove the lower bound, we present a PSSR scheme that generalizes the MPIR scheme of~\cite{HWS2025}, in which all subsets of a fixed size are candidate demands and the scheme is message-wise symmetric.
In contrast, our scheme is tailored to the structure of the demand subsets and allows the query and answer structure to vary across messages. 
This yields a unified framework for both full and structured demand families.
Specifically, we introduce a subclass of balanced ${\{0,1\}}$-linear PSSR schemes parameterized by a set of variables, and develop an optimization framework to select these parameters so as to maximize the retrieval rate subject to the correctness and privacy conditions.

\begin{theorem}
\label{thm:PRSR_L*_bounds}
For $N$ servers and $E$ candidate demand index sets $\mathrm{W}_1,\dots,\mathrm{W}_E$, the minimum subpacketization level required by any balanced $\{0,1\}$-linear PSSR scheme achieving the rate $R_{*}$ is lower bounded by  
\begin{equation}\label{eq:L_lb}
L_{*}\coloneqq \frac{N\alpha}{\gcd(N\alpha,D\beta)},
\end{equation} where $\alpha$ and $\beta$ are coprime integers with $R_{*} = \alpha/\beta$,
and is upper bounded by 
\begin{equation}
\label{eq:L_objc}
L^{*} \coloneqq \min L,
\end{equation}
where the minimization is over the same variables as in Theorem~\ref{thm:PRSR capacity}, namely $\{T_{\mathrm{U}}\}$, $\{I_{\mathrm{U},\mathrm{V}}^{[\mathrm{W}_j]}\}$, $\{J^{[\mathrm{W}_j]}_{\mathrm{V},i}(k)\}$, and $L$, subject to the constraints~\eqref{eq:set_1_cons}--\eqref{eq:L_geq_1} and the additional constraint
\begin{equation}
    \label{eq:min_L_cons}
  \frac{1}{L}  \sum_{\mathrm{U} \subseteq [1:K]} T_{\mathrm{U}} = \frac{D}{NR_{*}}.
\end{equation}
\end{theorem}

The lower bound follows because a balanced scheme must retrieve an integer number of symbols from each server. 
To prove the upper bound, we refine the optimization framework used to characterize the rate $R_{*}$. 
We retain the same class of schemes and the same set of constraints, but minimize the subpacketization level subject to achieving rate $R_{*}$.

\begin{remark}
\normalfont
For a given demand family, the PSSR rate upper bound $R^{*}$ can be larger than the best-known MPIR rate upper bound, which was originally established in~\cite{BU2018} and is also recovered from~\eqref{eq:PSR ub} by specializing PSSR to the full demand family.
As a result, whenever $R^{*}$ exceeds the MPIR bound and is achievable, rate-optimal MPIR schemes become suboptimal for the corresponding PSSR instance.
\end{remark}

\begin{remark}
\normalfont
The achievable PSSR rate $R_{*}$ is no smaller than the best-known retrieval rate achieved by existing ${\{0,1\}}$-linear MPIR schemes~\cite{BU2018,HWS2025}. 
Indeed, any such scheme remains valid upon restricting the demand space from the full family to an arbitrary subfamily, and thus induces a feasible solution to our optimization problem.
However, these MPIR-induced solutions need not be rate-optimal for a given PSSR instance; by exploiting the structure of the candidate demands, our optimization can instead achieve strictly higher retrieval rates.
\end{remark}

\begin{remark}
\normalfont
Demand structure can be beneficial even when it does not improve the optimal retrieval rate. 
Indeed, restricting the demand space to a structured family may leave the optimal retrieval rate unchanged relative to the full demand family, but a scheme that is rate-optimal for the full family (and hence remains rate-optimal under the restriction) may still require a large subpacketization level.
In contrast, our optimization framework treats the subpacketization level as an explicit design parameter: 
after imposing rate-optimality within the considered class of schemes, it allows us to choose a PSSR scheme with the smallest subpacketization level in that class, rather than inheriting the fixed subpacketization level that is implicit in existing MPIR schemes.
\end{remark}

\begin{remark}
\normalfont
Our parallel work~\cite{IH2026PCBRarXiv} applies the proposed optimization framework to the contiguous-block demand family, consisting of all ${K-D+1}$ contiguous length-${D}$ message blocks without wrap-around.
For this family, we show that the resulting scheme is rate-optimal and, for a broad range of parameter values, has optimal subpacketization matching the corresponding lower bound.
Specifically, relative to the best-known balanced ${\{0,1\}}$-linear MPIR schemes, when $D \nmid K$ our scheme is rate-optimal whereas the MPIR schemes are rate-suboptimal; when $D \mid K$, both are rate-optimal, but our scheme requires subpacketization level $N^{K/D}$, whereas the MPIR schemes require subpacketization level at least $N^{K-D+1}/D$, which can be much larger; see~\cite{H2026} for details.
\end{remark}

\begin{remark}
\normalfont
The tightness of our converse bounds---the rate upper bound $R^{*}$ and the subpacketization-level lower bound $L_{*}$---remains open in general, even in the MPIR setting, which is a special case of PSSR. 
Likewise, the optimality of our scheme---which yields the rate lower bound $R_{*}$ and the subpacketization-level upper bound $L^{*}$---remains open in general, both in terms of retrieval rate and subpacketization level, even for highly structured and symmetric demand families, such as those induced by regular graphs when ${D=2}$ (and by regular uniform hypergraphs when ${D>2}$).
\end{remark}

\begin{remark}
\normalfont
Since both the converse bound and the achievable rate depend on the underlying family of candidate demands, a universal closed-form characterization for all instances is generally out of reach.
Moreover, even for a fixed demand family, the relevant optimization problems can involve a space whose size grows exponentially in the total number of messages $K$ and the number of demand messages $D$.
However, this size can often be reduced by exploiting automorphisms of the demand family.
Instead of treating every subset of messages separately, one can group equivalent subsets into orbits under the induced automorphism action.
Working with one representative from each orbit collapses identical variables and constraints, thereby reducing the size of the resulting optimization problems.
\end{remark}

\section{Converse Proofs}

In this section, we establish our converse results: the upper bound on the retrieval rate, $R^{*}$, in Theorem~\ref{thm:PRSR capacity} and the lower bound on the subpacketization level, $L_{*}$, in Theorem~\ref{thm:PRSR_L*_bounds}.

\subsection{Upper Bounding the Achievable Rate}
\label{sec:conv_proof_PSSR}

The proof relies on the following two lemmas, which are analogous to \cite[Lemma~1]{WBU2022} and \cite[Lemma~2]{WBU2022}, and which we include here for completeness.

\begin{lemma}
    \label{lemma: symmetry}
    For any $\mathrm{W},\mathrm{W}'\in \{\mathrm{W}_1,\dots,\mathrm{W}_E\}$, any ${\mathrm{U}\subseteq[1:K]}$, and any $n\in[1:N]$, it holds that
    \begin{align}
         H(\mathbf{A}_{n}^{[\mathrm{W}]} | \mathbf{X}_{\mathrm{U}}, \mathbf{Q}_{n}^{[\mathrm{W}]}) & = H(\mathbf{A}_{n}^{[\mathrm{W}']} | \mathbf{X}_{\mathrm{U}}, \mathbf{Q}_{n}^{[\mathrm{W}']}).  \label{eq: symmetry_lemma_1}
    \end{align}
\end{lemma}
\begin{proof}
From the privacy condition in~\eqref{eq:privacy}, we have    
    \begin{align}
         H(\mathbf{A}_{n}^{[\mathrm{W}]},  \mathbf{Q}_{n}^{[\mathrm{W}]}, \mathbf{X}_{\mathrm{U}}) & = H(\mathbf{A}_{n}^{[\mathrm{W}']},  \mathbf{Q}_{n}^{[\mathrm{W}']}, \mathbf{X}_{\mathrm{U}}),  \label{eq:user_priv_1}\\   
         H(\mathbf{Q}_{n}^{[\mathrm{W}]}, \mathbf{X}_{\mathrm{U}}) & = H(\mathbf{Q}_{n}^{[\mathrm{W}']}, \mathbf{X}_{\mathrm{U}}). \label{eq:user_priv_2}
    \end{align}
Applying the chain rule to~\eqref{eq:user_priv_1} and using~\eqref{eq:user_priv_2} yields~\eqref{eq: symmetry_lemma_1}, completing the proof. 
\end{proof}

\begin{lemma}
\label{lemma: given_Q1:N_=_given_Qn}
For any $\mathrm{W} \in \{\mathrm{W}_1,\dots,\mathrm{W}_E\}$, any $\mathrm{U} \subseteq [1:K]$, and any $n \in [1:N]$, it holds that
\begin{equation}\label{eq:lemma2}
    H(\mathbf{A}_{n}^{[\mathrm{W}]} | \mathbf{Q}_{[1:N]}^{[\mathrm{W}]}, \mathbf{X}_{\mathrm{U}}) = H(\mathbf{A}_{n}^{[\mathrm{W}]} | \mathbf{Q}_{n}^{[\mathrm{W}]}, \mathbf{X}_{\mathrm{U}}).
\end{equation}
\end{lemma}

\begin{proof} Since conditioning cannot increase entropy, we have 
\[
H(\mathbf{A}_{n}^{[\mathrm{W}]} | \mathbf{Q}_{n}^{[\mathrm{W}]}, \mathbf{X}_{\mathrm{U}} )\geq H ( \mathbf{A}_{n}^{[\mathrm{W}]} | \mathbf{Q}_{[1:N]}^{[\mathrm{W}]}, \mathbf{X}_{\mathrm{U}} ).
\]
To prove~\eqref{eq:lemma2}, it thus suffices to show the reverse inequality, 
\[
H(\mathbf{A}_{n}^{[\mathrm{W}]} | \mathbf{Q}_{n}^{[\mathrm{W}]}, \mathbf{X}_{\mathrm{U}} )\leq H ( \mathbf{A}_{n}^{[\mathrm{W}]} | \mathbf{Q}_{[1:N]}^{[\mathrm{W}]}, \mathbf{X}_{\mathrm{U}} ),
\]
which follows from the chain below: 
    \begin{align}
    & H(\mathbf{A}_{n}^{[\mathrm{W}]} | \mathbf{Q}_{n}^{[\mathrm{W}]}, \mathbf{X}_{\mathrm{U}} ) - H ( \mathbf{A}_{n}^{[\mathrm{W}]} | \mathbf{Q}_{[1:N]}^{[\mathrm{W}]}, \mathbf{X}_{\mathrm{U}} ) \nonumber 
    \\ & \quad = I ( \mathbf{A}_{n}^{[\mathrm{W}]} ; \mathbf{Q}_{[1:N]}^{[\mathrm{W}]} | \mathbf{Q}_{n}^{[\mathrm{W}]}, \mathbf{X}_{\mathrm{U}} ) \label{eq:I_An_Q1:N_given_Qn_Xw} 
    \\ & \quad \leq I ( \mathbf{A}_{n}^{[\mathrm{W}]}, \mathbf{X}_{[1:K]}; \mathbf{Q}_{[1:N]}^{[\mathrm{W}]} | \mathbf{Q}_{n}^{[\mathrm{W}]}, \mathbf{X}_{\mathrm{U}} ) 
    \label{eq:monotonicity_MI}
    \\ & \quad = I ( \mathbf{X}_{[1:K]}; \mathbf{Q}_{[1:N]}^{[\mathrm{W}]} | \mathbf{Q}_{n}^{[\mathrm{W}]}, \mathbf{X}_{\mathrm{U}} ) \nonumber \\ & \quad \quad + I ( \mathbf{A}_{n}^{[\mathrm{W}]} ; \mathbf{Q}_{[1:N]}^{[\mathrm{W}]} | \mathbf{Q}_{n}^{[\mathrm{W}]}, \mathbf{X}_{[1:K]}) \label{eq:I_An_Q1:N_given_Qn_Xw_chain_rule} 
    \\ & \quad = I ( \mathbf{X}_{[1:K]}; \mathbf{Q}_{[1:N]}^{[\mathrm{W}]} | \mathbf{Q}_{n}^{[\mathrm{W}]}, \mathbf{X}_{\mathrm{U}} ) \nonumber \\ 
    & \quad \quad  + H ( \mathbf{A}_{n}^{[\mathrm{W}]} | \mathbf{Q}_{n}^{[\mathrm{W}]}, \mathbf{X}_{[1:K]}) - H ( \mathbf{A}_{n}^{[\mathrm{W}]} | \mathbf{Q}_{[1:N]}^{[\mathrm{W}]}, \mathbf{X}_{[1:K]}) \label{eq:I_X_S_Q1:N_+H_An_-H_An} 
    \\ & \quad = I ( \mathbf{X}_{[1:K]}; \mathbf{Q}_{[1:N]}^{[\mathrm{W}]} | \mathbf{Q}_{n}^{[\mathrm{W}]}, \mathbf{X}_{\mathrm{U}} ) \label{eq:I_X_S_Q1:N_given_Qn_Xw} 
    \\ & \quad \leq I ( \mathbf{X}_{[1:K]}; \mathbf{Q}_{[1:N]}^{[\mathrm{W}]} | \mathbf{Q}_{n}^{[\mathrm{W}]}, \mathbf{X}_{\mathrm{U}} ) + I ( \mathbf{X}_{\mathrm{U}}; \mathbf{Q}_{[1:N]}^{[\mathrm{W}]} | \mathbf{Q}_{n}^{[\mathrm{W}]}  ) 
    \label{eq:nonnegativity_MI_1} \\ & \quad = I ( \mathbf{X}_{[1:K]}; \mathbf{Q}_{[1:N]}^{[\mathrm{W}]} | \mathbf{Q}_{n}^{[\mathrm{W}]} ) \label{eq:I_X_S_Q1:N_given_Qn} 
    \\ & \quad \leq I ( \mathbf{X}_{[1:K]}; \mathbf{Q}_{[1:N]}^{[\mathrm{W}]} | \mathbf{Q}_{n}^{[\mathrm{W}]} ) + I ( \mathbf{X}_{[1:K]}; \mathbf{Q}_{n}^{[\mathrm{W}]} ) \label{eq:nonnegativity_MI_2}
    \\ & \quad = I ( \mathbf{X}_{[1:K]}; \mathbf{Q}_{[1:N]}^{[\mathrm{W}]} ) \label{eq:I_X_S_Q1:N} 
    \\ & \quad = 0 \label{eq:I_X_S_Q1:N_equal_0}.
    \end{align}
Here,~\eqref{eq:I_An_Q1:N_given_Qn_Xw} and~\eqref{eq:I_X_S_Q1:N_+H_An_-H_An} follow from the definition of mutual information; 
\eqref{eq:monotonicity_MI} follows from the monotonicity of mutual information;  
\eqref{eq:I_An_Q1:N_given_Qn_Xw_chain_rule},~\eqref{eq:I_X_S_Q1:N_given_Qn}, and~\eqref{eq:I_X_S_Q1:N} follow from the chain rule; 
\eqref{eq:I_X_S_Q1:N_given_Qn_Xw} follows from~\eqref{eq:A_func_Q_X_S}; 
\eqref{eq:nonnegativity_MI_1} and~\eqref{eq:nonnegativity_MI_2} follow from the non-negativity of mutual information; 
and~\eqref{eq:I_X_S_Q1:N_equal_0} follows from \eqref{eq:Q_indep_X_S}.
\end{proof}

We now prove the rate upper bound in Theorem~\ref{thm:PRSR capacity}.
Recall that the retrieval rate, defined in~\eqref{eq:RateDef}, is the ratio between $H(\mathbf{X}_{\mathbf{W}})$ and $\sum_{n=1}^N H(\mathbf{A}^{[\mathbf{W}]}_n|\mathbf{Q}^{[\mathbf{W}]}_n)$. 
First, note that 
\begin{equation}\label{eq:H_Xw}
H(\mathbf{X}_{\mathbf{W}}) = DL.
\end{equation} 
This is because: (i) $H(\mathbf{X}_{\mathbf{W}})\geq H(\mathbf{X}_{\mathbf{W}}|\mathbf{W}) = DL$, since ${\mathbf{X}_1,\dots,\mathbf{X}_K}$ are independent and uniformly distributed over $\mathbb{F}_q^{L}$, which implies $H(\mathbf{X}_{\mathrm{W}_j}) = DL$ for every ${j\in [1:E]}$; 
and $H(\mathbf{X}_{\mathbf{W}})\leq \log_q |\mathbb{F}_q^{DL}| = DL$, since $\mathbf{X}_{\mathbf{W}}$ takes values in $\mathbb{F}_q^{DL}$.   
To upper bound the rate, it thus remains to lower bound $\sum_{n=1}^N H(\mathbf{A}^{[\mathbf{W}]}_n|\mathbf{Q}^{[\mathbf{W}]}_n)$. 
To do this, we write
\begin{align}
  & \sum_{n=1}^N  H(\mathbf{A}_{n}^{[\mathbf{W}]} | \mathbf{Q}_{n}^{[\mathbf{W}]}) \nonumber \\ 
  & \quad \geq \sum_{n=1}^{N} H(\mathbf{A}_{n}^{[\mathbf{W}]} | \mathbf{Q}_{n}^{[\mathbf{W}]},\mathbf{W})\label{eq:cond_on_W} \\ 
  & \quad  = \sum_{n=1}^{N} H(\mathbf{A}_{n}^{[\mathrm{W}_{1}]} | \mathbf{Q}_{n}^{[\mathrm{W}_1]}) \label{eq:remove_cond_on_W} \\
  & \quad =  \sum_{n=1}^N H (\mathbf{A}_{n}^{[\mathrm{W}_1]} | \mathbf{Q}_{[1:N]}^{[\mathrm{W}_1]}) \label{eq: sum_H_An_given_Q1:N} \\ 
  & \quad \geq  H(\mathbf{A}_{[1:N]}^{[\mathrm{W}_1]} |  \mathbf{Q}_{[1:N]}^{[\mathrm{W}_1]}), \label{eq: H_A1:N_given_Q1:N}
\end{align} 
where~\eqref{eq:cond_on_W} follows since conditioning cannot increase entropy; \eqref{eq:remove_cond_on_W} follows from Lemma~\ref{lemma: symmetry}, which implies that $H(\mathbf{A}_{n}^{[\mathrm{W}_j]} | \mathbf{Q}_{n}^{[\mathrm{W}_j]}) = H(\mathbf{A}_{n}^{[\mathrm{W}_1]} | \mathbf{Q}_{n}^{[\mathrm{W}_1]})$ for all ${j\in[1:E]}$, noting that any $\mathrm{W}_{k_1}$, ${k_1\in [1:E]}$, may be used in place of $\mathrm{W}_1$; 
\eqref{eq: sum_H_An_given_Q1:N} follows from Lemma~\ref{lemma: given_Q1:N_=_given_Qn}; 
and~\eqref{eq: H_A1:N_given_Q1:N} follows from the subadditivity of entropy.

Using the chain rule in two different orders, we have
\begin{align}
 & H(\mathbf{A}_{[1:N]}^{[\mathrm{W}_{1}]},\mathbf{X}_{\mathrm{W}_{1}} | \mathbf{Q}_{[1:N]}^{[\mathrm{W}_{1}]}) \nonumber   
 \\ & \quad =   H(\mathbf{A}_{[1:N]}^{[\mathrm{W}_1]} | \mathbf{Q}_{[1:N]}^{[\mathrm{W}_1]})+  H(\mathbf{X}_{\mathrm{W}_1} | \mathbf{A}_{[1:N]}^{[\mathrm{W}_1]}, \mathbf{Q}_{[1:N]}^{[\mathrm{W}_1]}) \label{eq:chain_rule_A_X_1}
 \\ & \quad = H(\mathbf{X}_{\mathrm{W}_1} | \mathbf{Q}_{[1:N]}^{[\mathrm{W}_1]}) + H(\mathbf{A}_{[1:N]}^{[\mathrm{W}_1]} | \mathbf{Q}_{[1:N]}^{[\mathrm{W}_1]},\mathbf{X}_{\mathrm{W}_1} ) \label{eq:chain_rule_A_X_2}
\end{align}
Equating~\eqref{eq:chain_rule_A_X_1} and~\eqref{eq:chain_rule_A_X_2} yields
\begin{align}
 & H(\mathbf{A}_{[1:N]}^{[\mathrm{W}_1]} | \mathbf{Q}_{[1:N]}^{[\mathrm{W}_1]} ) \nonumber 
 \\ & \quad =  H(\mathbf{X}_{\mathrm{W}_1} | \mathbf{Q}_{[1:N]}^{[\mathrm{W}_1]} )+  H(\mathbf{A}_{[1:N]}^{[\mathrm{W}_1]} | \mathbf{Q}_{[1:N]}^{[\mathrm{W}_1]},\mathbf{X}_{\mathrm{W}_1} ) \nonumber 
 \\ & \quad \quad - H(\mathbf{X}_{\mathrm{W}_1} | \mathbf{A}_{[1:N]}^{[\mathrm{W}_1]},\mathbf{Q}_{[1:N]}^{[\mathrm{W}_1]})
 \nonumber \\ & \quad = H(\mathbf{X}_{\mathrm{W}_1}) + H(\mathbf{A}_{[1:N]}^{[\mathrm{W}_1]} | \mathbf{Q}_{[1:N]}^{[\mathrm{W}_1]},\mathbf{X}_{\mathrm{W}_1} ), \label{eq:A_given_Q_2}
\end{align}
where~\eqref{eq:A_given_Q_2} follows from~\eqref{eq:Q_indep_X_S} and~\eqref{eq:correctness}.

Combining~\eqref{eq: H_A1:N_given_Q1:N} and~\eqref{eq:A_given_Q_2}, we have
\begin{align}
& \sum_{n=1}^N H(\mathbf{A}_{n}^{[\mathbf{W}]} | \mathbf{Q}_{n}^{[\mathbf{W}]}) \nonumber \\ 
& \quad \geq H(\mathbf{X}_{\mathrm{W}_1} ) + H(\mathbf{A}_{[1:N]}^{[\mathrm{W}_1]} | \mathbf{Q}_{[1:N]}^{[\mathrm{W}_1]},\mathbf{X}_{\mathrm{W}_1} ).
\label{eq: H_Xw1_plus_H_A1:N_given_Q1:N_Xw1}
\end{align}
To further lower bound $H(\mathbf{A}_{[1:N]}^{[\mathrm{W}_1]} |  \mathbf{Q}_{[1:N]}^{[\mathrm{W}_1]}, 
\mathbf{X}_{\mathrm{W}_1})$, we write 
\begin{align}
  & H(\mathbf{A}_{[1:N]}^{[\mathrm{W}_1]} |  \mathbf{Q}_{[1:N]}^{[\mathrm{W}_1]},\mathbf{X}_{\mathrm{W}_1}) \nonumber \\ & \quad \geq  \frac{1}{N} \sum_{n=1}^{N} H(\mathbf{A}_{n}^{[\mathrm{W}_1]} |  \mathbf{Q}_{[1:N]}^{[\mathrm{W}_1]},\mathbf{X}_{\mathrm{W}_1}) \label{eq: H_Xw1_plus_monotonicity}
    \\ & \quad = \frac{1}{N} \sum_{n=1}^{N} H(\mathbf{A}_{n}^{[\mathrm{W}_1]} |  \mathbf{Q}_{n}^{[\mathrm{W}_1]},\mathbf{X}_{\mathrm{W}_1}) \label{eq: H_Xw1_plus_monotonicity_with_lemma}
    \\ & \quad = \frac{1}{N} \sum_{n=1}^{N} H(\mathbf{A}_{n}^{[\mathrm{W}_2]} |  \mathbf{Q}_{n}^{[\mathrm{W}_2]},\mathbf{X}_{\mathrm{W}_1}) \label{eq: H_Xw1_plus_monotonicity_with_W2}
    \\ & \quad = \frac{1}{N} \sum_{n=1}^{N} H(\mathbf{A}_{n}^{[\mathrm{W}_2]} |  \mathbf{Q}_{[1:N]}^{[\mathrm{W}_2]},\mathbf{X}_{\mathrm{W}_1}) \label{eq: H_Xw1_plus_monotonicity_with_W2_given_Q1:N}
    \\ & \quad \geq \frac{1}{N} H(\mathbf{A}_{[1:N]}^{[\mathrm{W}_2]} | \mathbf{Q}_{[1:N]}^{[\mathrm{W}_2]},\mathbf{X}_{\mathrm{W}_1}) \label{eq: H_Xw1_plus_1/N_H_A1:N_given_Q1:N}
\end{align} 
where~\eqref{eq: H_Xw1_plus_monotonicity} follows from the monotonicity of entropy; 
\eqref{eq: H_Xw1_plus_monotonicity_with_lemma} and~\eqref{eq: H_Xw1_plus_monotonicity_with_W2_given_Q1:N} follow from Lemma~\ref{lemma: given_Q1:N_=_given_Qn}; 
\eqref{eq: H_Xw1_plus_monotonicity_with_W2} follows from Lemma~\ref{lemma: symmetry}, noting that any $\mathrm{W}_{k_2}$, $k_2\in [1:E]\setminus \{k_1\}$, may be used in place of $\mathrm{W}_2$; 
and~\eqref{eq: H_Xw1_plus_1/N_H_A1:N_given_Q1:N} follows from the subadditivity of entropy. 

Combining~\eqref{eq: H_Xw1_plus_H_A1:N_given_Q1:N_Xw1} and~\eqref{eq: H_Xw1_plus_1/N_H_A1:N_given_Q1:N}, we have
\begin{align}
& \sum_{n=1}^N H(\mathbf{A}_{n}^{[\mathbf{W}]} | \mathbf{Q}_{n}^{[\mathbf{W}]}) \nonumber \\ 
& \quad \geq H(\mathbf{X}_{\mathrm{W}_1}) + \frac{1}{N}H(\mathbf{A}_{[1:N]}^{[\mathrm{W}_2]} | \mathbf{Q}_{[1:N]}^{[\mathrm{W}_2]},\mathbf{X}_{\mathrm{W}_1}).
\label{eq:H_Xw1_plus_1/N_H_A1:N_w2_given_Q1:N_W2_Xw1}
\end{align}

Applying steps analogous to those in~\eqref{eq:chain_rule_A_X_1}--\eqref{eq:A_given_Q_2} and~\eqref{eq: H_Xw1_plus_monotonicity}--\eqref{eq: H_Xw1_plus_1/N_H_A1:N_given_Q1:N}, it follows that, for any ${j\in[1:E-1]}$,
\begin{align}
\label{eq: H(A^2) lower bound}
 &  H(\mathbf{A}_{[1:N]}^{[\mathrm{W}_j]} | \mathbf{Q}_{[1:N]}^{[\mathrm{W}_j]},\mathbf{X}_{\mathrm{W}_{[0:j-1]}}) \nonumber \\ 
 & \quad \geq H(\mathbf{X}_{\mathrm{W}_j} | \mathbf{X}_{\mathrm{W}_{[0:j-1]}})+\frac{1}{N}H(\mathbf{A}_{[1:N]}^{[\mathrm{W}_{j+1}]} | \mathbf{Q}_{[1:N]}^{[\mathrm{W}_{j+1}]},\mathbf{X}_{\mathrm{W}_{[1:j]}}),
\end{align} 
where $\mathrm{W}_{[l:m]} \coloneqq \bigcup_{k=l}^{m} \mathrm{W}_{k}$ for ${l,m\in [0:E]}$ with ${l\leq m}$, and $\mathrm{W}_{0}\coloneqq \emptyset$. Note that for each ${j\in [1:E-1]}$, any $\mathrm{W}_{k_{j+1}}$, ${k_{j+1}\in [1:E]\setminus \{k_1,\dots,k_j\}}$, may be used in place of $\mathrm{W}_{j+1}$.  

Moreover, we have 
\begin{align}
& H(\mathbf{A}_{[1:N]}^{[\mathrm{W}_E]} | \mathbf{Q}_{[1:N]}^{[\mathrm{W}_E]},\mathbf{X}_{\mathrm{W}_{[1:E-1]}}) \nonumber \\
& \quad = H(\mathbf{A}_{[1:N]}^{[\mathrm{W}_E]} | \mathbf{Q}_{[1:N]}^{[\mathrm{W}_E]},\mathbf{X}_{\mathrm{W}_{[1:E-1]}}) \nonumber \\ 
& \quad \quad + H( \mathbf{X}_{\mathrm{W}_E} |\mathbf{A}_{[1:N]}^{[\mathrm{W}_E]}, \mathbf{Q}_{[1:N]}^{[\mathrm{W}_E]},\mathbf{X}_{\mathrm{W}_{[1:E-1]}}) \label{eq:H_XwF_given_A1:N_Q1:N_Xw1:F-1}\\
& \quad = H(\mathbf{X}_{\mathrm{W}_E}|\mathbf{Q}_{[1:N]}^{[\mathrm{W}_E]},\mathbf{X}_{\mathrm{W}_{[1:E-1]}}) \nonumber \\ 
& \quad \quad + H(\mathbf{A}_{[1:N]}^{[\mathrm{W}_E]}|\mathbf{Q}_{[1:N]}^{[\mathrm{W}_E]}, \mathbf{X}_{\mathrm{W}_{[1:E]}}) \label{eq:H_XwF_given_Q1:N_Xw1:F-1+H_A1:N_given_Q1:N_Xw1:F} \\ 
& \quad = H(\mathbf{X}_{\mathrm{W}_E}|\mathbf{Q}_{[1:N]}^{[\mathrm{W}_E]},\mathbf{X}_{\mathrm{W}_{[1:E-1]}}) \label{eq:H_XwF_given_Q1:N_Xw1:F-1} \\ 
& \quad = H(\mathbf{X}_{\mathrm{W}_E}|\mathbf{X}_{\mathrm{W}_{[1:E-1]}}), \label{eq:H_XwF_given_Xw1:F-1}
\end{align} 
where~\eqref{eq:H_XwF_given_A1:N_Q1:N_Xw1:F-1} follows from~\eqref{eq:correctness}; 
\eqref{eq:H_XwF_given_Q1:N_Xw1:F-1+H_A1:N_given_Q1:N_Xw1:F} follow from the chain rule;
\eqref{eq:H_XwF_given_Q1:N_Xw1:F-1} follows from~\eqref{eq:A_func_Q_X_S}; 
and~\eqref{eq:H_XwF_given_Xw1:F-1} follows from~\eqref{eq:Q_indep_X_S}.  

Combining~\eqref{eq:H_Xw1_plus_1/N_H_A1:N_w2_given_Q1:N_W2_Xw1},~\eqref{eq: H(A^2) lower bound}, and~\eqref{eq:H_XwF_given_Xw1:F-1} yields
\begin{align}
 \sum_{n=1}^N H (\mathbf{A}_{n}^{[\mathbf{W}]} | \mathbf{Q}_{n}^{[\mathbf{W}]}) & \geq \sum_{j=1}^{E} \frac{1}{N^{j-1}} H(\mathbf{X}_{\mathrm{W}_{j}} |\mathbf{X}_{\mathrm{W}_{[0:j-1]}}) \nonumber
 \\ & = L
 \sum_{j=1}^{E} \frac{1}{N^{j-1}} \left| {\mathrm{W}_{j}} \setminus \textstyle\bigcup_{k=0}^{j-1}\mathrm{W}_{k} \right|, \label{eq: DC final lower bound}
\end{align}
where~\eqref{eq: DC final lower bound} follows because ${\mathbf{X}_1,\dots,\mathbf{X}_K}$ are independent and uniformly distributed over $\mathbb{F}_q^{L}$, and hence, for any  ${\mathrm{U}_1,\mathrm{U}_2\subseteq [1:K]}$, we have \[H(\mathbf{X}_{\mathrm{U}_1}|\mathbf{X}_{\mathrm{U}_2}) = H(\mathbf{X}_{\mathrm{U}_1\setminus \mathrm{U}_2}|\mathbf{X}_{\mathrm{U}_2}) = H(\mathbf{X}_{\mathrm{U}_1\setminus \mathrm{U}_2}) = |\mathrm{U}_1\setminus \mathrm{U}_2|L.\]

Fix an arbitrary permutation $\pi: [1:E]\rightarrow [1:E]$. 
By applying the same arguments as above, but ordering the candidate demand index sets as $\mathrm{W}_{\pi(1)},\dots,\mathrm{W}_{\pi(E)}$ (i.e., taking $k_j = \pi(j)$ for all ${j\in [1:E]}$) instead of $\mathrm{W}_1,\dots,\mathrm{W}_E$, we get
\begin{align}\label{eq:DC_LB_perm}
& \sum_{n=1}^N H (\mathbf{A}_{n}^{[\mathbf{W}]} | \mathbf{Q}_{n}^{[\mathbf{W}]}) \nonumber \\ 
& \quad \geq L
 \sum_{j=1}^{E} \frac{1}{N^{j-1}} \left| {\mathrm{W}_{\pi(j)}} \setminus \textstyle\bigcup_{k=0}^{j-1} {\mathrm{W}_{\pi(k)}} \right|,
\end{align} where $\pi(0)\coloneqq 0$. 

Since~\eqref{eq:DC_LB_perm} holds for all permutations $\pi$, we have
\begin{align}\label{eq:DC_LB_perm_max}
& \sum_{n=1}^N H (\mathbf{A}_{n}^{[\mathbf{W}]} | \mathbf{Q}_{n}^{[\mathbf{W}]}) \nonumber \\ 
& \quad \geq L \left(\max
 \sum_{j=1}^{E} \frac{1}{N^{j-1}} \left| {\mathrm{W}_{\pi(j)}} \setminus \textstyle\bigcup_{k=0}^{j-1} {\mathrm{W}_{\pi(k)}} \right|\right),
\end{align} where the maximization is over all permutations $\pi$. 

Combining~\eqref{eq:H_Xw} and~\eqref{eq:DC_LB_perm_max} yields the upper bound $R^{*}$ in~\eqref{eq:PSR ub}. 

\subsection{Lower Bounding the Subpacketization Level}

Consider an arbitrary balanced ${\{0,1\}}$-linear scheme that achieves rate $R_{*}$. 
Such a scheme involves retrieving a total of ${DL/R_{*}}$ linear combinations of message subpackets, and these combinations are distributed evenly across the $N$ servers. 
Consequently, the user retrieves ${DL/(NR_{*})}$ linear combinations from each server, and this quantity must be an integer.

Equivalently, the subpacketization level $L$ must be chosen so that ${DL/(NR_{*})}$ is an integer, and the smallest integer $L$ for which this holds is $L_{*}$, as defined in~\eqref{eq:L_lb}.

\section{Achievability Proofs}\label{sec:Ach_Proofs}

In this section, we prove the achievability of the retrieval-rate lower bound $R_{*}$ in Theorem~\ref{thm:PRSR capacity} and the subpacketization-level upper bound $L^{*}$ in Theorem~\ref{thm:PRSR_L*_bounds}.

\subsection{Scheme Description}

We consider balanced $\{0,1\}$-linear PSSR schemes, in which each message is partitioned into $L$ subpackets, each a randomly chosen $\mathbb{F}_q$-symbol from the message. 
The user retrieves from each server a collection of \emph{symbols}, each a $\{0,1\}$-linear combination of message subpackets. 
Accordingly, for each server $n\in[1:N]$, the user's query $\mathrm{Q}^{[\mathrm{W}]}_n$ specifies the requested symbols by listing, for each symbol, its support and the subpacket indices of every participating message; 
the server's answer $\mathrm{A}^{[\mathrm{W}]}_n$ is the corresponding list of returned symbols obtained by forming the specified $\{0,1\}$-linear combinations.

We focus on a subclass of balanced $\{0,1\}$-linear schemes that satisfies the following four structural properties, which make the schemes easier to analyze and simplify the encoding and decoding procedures:
\begin{itemize}
\item For every fixed set of messages, each server is queried for the same number of symbols, i.e., $\{0,1\}$-linear combinations, involving subpackets of those messages;
\item For every symbol retrieved from any server, each message contributes at most one of its subpackets;
\item For every subpacket and every server, that subpacket appears in at most one symbol retrieved from that server; 
\item For every demand message, the number of its recovered subpackets is the same across all servers, and consequently the number of subpackets per message, $L$, must be divisible by the number of servers, $N$.
\end{itemize}

We parametrize such schemes by integer variables $T_{\mathrm{U}}$ indexed by non-empty subsets $\mathrm{U}\subseteq [1:K]$, where $T_{\mathrm{U}}$ denotes the number of symbols per server whose support, i.e., the set of participating messages, is exactly $\mathrm{U}$. 
In particular, $T_{\{i\}}$ is the number of singletons of message $\mathrm{X}_i$,
$T_{\{i_1,i_2\}}$ is the number of two-message sums involving messages $\mathrm{X}_{i_1}$ and $\mathrm{X}_{i_2}$, and similarly for larger supports.

Within this subclass, we characterize a scheme that maximizes the achievable rate over all $L$ divisible by $N$, and further identify a scheme that achieves this maximum rate with the smallest such $L$.

Throughout the remainder of this section, we fix an arbitrary candidate demand index set $\mathrm{W}_j$ with ${j\in [E]}$ and analyze the scheme for this representative choice.

\subsection{Side--Target Pairings}

A symbol with support $\mathrm{U}\subseteq \mathrm{W}_j$ involves only subpackets of demand messages;
such a \emph{demand-only} symbol can be used directly to recover a new demand subpacket, as detailed later. 
In contrast, if $\mathrm{U}\not\subseteq \mathrm{W}_j$, the symbol includes interference subpackets and therefore cannot be used directly to recover new demand subpackets.
Instead, such symbols are incorporated via cross-server cancellations to generate new
\emph{demand-only} symbols.
Specifically, every symbol whose support is not contained in $\mathrm{W}_j$ is assigned one of two roles: \emph{side} or \emph{target}.
For any such symbol from a given server, if it is designated as \emph{side}, then ${N-1}$ \emph{target} symbols are selected from the other servers (one per server), all with a common support that contains the side support and includes no interference messages beyond those already present in the side;
subtracting the side symbol from each selected target symbol yields a new \emph{demand-only} symbol.
If it is designated as \emph{target}, then it is one of these selected target symbols: it is paired with a \emph{side} symbol from another server, and subtracting that side symbol yields the corresponding \emph{demand-only} symbol.

To track these side--target pairings, we introduce integer variables
$I^{[\mathrm{W}_{j}]}_{\mathrm{U},\mathrm{V}}$, indexed by $\mathrm{U}\not\subseteq \mathrm{W}_j$ and $\mathrm{V}\subseteq \mathrm{W}_j\setminus \mathrm{U}$, which count, per server, the number of such pairings in which a side symbol supported on $\mathrm{U}$
is subtracted from a target symbol supported on $\mathrm{U}\cup \mathrm{V}$ from another server, yielding a new demand-only symbol supported on $\mathrm{V}$.

When generating a demand-only symbol via a side--target pairing, the side symbol may include demand messages, and all paired target symbols share those same messages.
Under the constraint that each message subpacket may appear at most once among the symbols the user retrieves from a server, reusing the \emph{same} demand subpacket in both the side symbol and its paired target symbols would exhaust that subpacket and leave no server from which it can be recovered.

Accordingly, whenever a demand message appears in both the side and target supports of a pairing, its occurrence in the side symbol is assigned a subpacket recovered earlier from a target server---not previously used at the side server,
whereas its occurrences in the paired target symbols are assigned a common subpacket (shared across the target symbols) recovered earlier from the side
server---not previously used at any target server.
This allows the cancellation to proceed without consuming any as-yet-unrecovered subpacket of the underlying demand message.

\subsection{Recovery Rounds}

A single-demand symbol involves exactly one demand message, and it directly reveals a new subpacket of that message.
A multi-demand symbol involves multiple demand messages and can likewise be used to recover a new subpacket of any one of those messages. 
This requires that every other demand subpacket appearing in the symbol was previously recovered from servers other than the one providing the symbol and can therefore be used for cancellation during recovery.

To enable these successive cancellations, demand subpackets are recovered over $D$ rounds: each round ${k\in[1:D]}$ involves demand-only symbols retrieved directly that contain exactly $k$ demand messages, and side--target pairings---yielding new demand-only symbols---whose target symbols contain exactly $k$ demand messages (so the corresponding side symbols contain at most ${k-1}$ demand messages).

Tracking how multi-demand symbols are used for recovery across messages and rounds, we introduce integer variables $J^{[\mathrm{W}_{j}]}_{\mathrm{V},i}(k)$, indexed by ${\mathrm{V}\subseteq \mathrm{W}_j}$ with ${|\mathrm{V}|\geq 2}$, ${i\in\mathrm{V}}$, and ${k\in[|\mathrm{V}|:D]}$, which count, per server, the number of demand-only symbols supported on $\mathrm{V}$ that are used in round $k$ to recover a new subpacket of message $\mathrm{X}_i$.

\subsection{Subpacket Indexing}
\label{sec: Subpacket Indexing}

The variables $\{T_{\mathrm U}\}$, $\{I^{[\mathrm W_j]}_{\mathrm U,\mathrm V}\}$, and $\{J^{[\mathrm W_j]}_{\mathrm V,i}(k)\}$ determine \emph{how many} symbols of each support are retrieved and how they are used for recovery, but they do not yet specify \emph{which} subpackets appear in each symbol. 

To complete the scheme description, we now describe the indexing of the message subpackets appearing in the symbols the user retrieves from each server. 

Subpacket indices are assigned over $D$ rounds, following the recovery process for all symbols involved in recovery; 
any remaining symbols are indexed afterward.

\subsubsection*{Round~1}
In Round~1, subpacket indices are assigned for all symbols containing at most one demand message, i.e., those with
supports $\mathrm{U}$ satisfying ${|\mathrm{U}\cap \mathrm{W}_j|\leq 1}$.
This includes:
(i) all singleton symbols, and
(ii) all side--target pairings whose targets contain exactly one demand message.

\begin{enumerate}
\item \emph{Singleton symbols.}
For each message $\mathrm X_i$, distinct unused subpacket indices are assigned to the $T_{\{i\}}$ singleton symbols of $\mathrm X_i$ from each server.
\item \emph{Side--target pairings whose targets contain one demand.}
Fix a pair ${(\mathrm{U},i)}$ with ${\mathrm{U}\subseteq [1:K]\setminus \mathrm{W}_j}$ and ${i\in \mathrm{W}_j}$.
For each side--target pairing with side support $\mathrm{U}$ and target support ${\mathrm{U}\cup\{i\}}$ used to recover ${N-1}$ subpackets of the message $\mathrm{X}_i$, subpacket indices are assigned as follows: 
\begin{itemize}
\item In the side symbol, all messages are assigned subpacket indices not previously used at any server;
\item In each target symbol, all messages in $\mathrm{U}$ are assigned the \emph{same} subpacket indices as in the side symbol, while the message $\mathrm{X}_i$ is assigned a subpacket index not previously used at any server, chosen distinctly across the target symbols.
\end{itemize}
\end{enumerate}

\subsubsection*{Round~$k$ (${k\in[2:D]}$)}
In Round~$k$, subpacket indices are assigned for all symbols used for recovery in that round. 
This includes:
(i) directly retrieved demand-only symbols containing exactly $k$ demand messages, and
(ii) side--target pairings whose targets contain exactly $k$ demand messages.

\begin{enumerate}
\item \emph{Directly retrieved demand-only symbols with $k$ demands.}
Fix a pair ${(\mathrm{V},i)}$ with ${\mathrm{V}\subseteq \mathrm{W}_j}$, ${|\mathrm{V}|=k}$, and ${i\in \mathrm{V}}$.
For each symbol supported on $\mathrm{V}$ from a fixed server that is used in Round~$k$ to recover a subpacket of
the message $\mathrm{X}_i$, subpacket indices are assigned as follows:
\begin{itemize}
\item The message $\mathrm{X}_i$ is assigned a subpacket index not previously used at any server.
\item For each message in ${\mathrm{V}\setminus\{i\}}$, a subpacket index not previously used at the given server is assigned; 
this index corresponds to a subpacket recovered in an earlier round from a different server.
\end{itemize}

\item \emph{Side--target pairings whose targets contain $k$ demands.}
Fix a triple ${(\mathrm{U},\mathrm{V},i)}$ with ${\mathrm{U}\not\subseteq \mathrm{W}_j}$, ${\mathrm{V}\subseteq \mathrm{W}_j\setminus \mathrm{U}}$, ${|(\mathrm{U}\cup \mathrm{V})\cap \mathrm{W}_j|=k}$, and ${i\in \mathrm{V}}$. 
For each side--target pairing with side support $\mathrm{U}$ and target support ${\mathrm{U}\cup \mathrm{V}}$ used in Round~$k$ to recover ${N-1}$ subpackets of the message $\mathrm{X}_i$, subpacket indices are assigned as follows:
\begin{itemize}
\item In the side symbol,
\begin{itemize}
\item For each message in ${\mathrm{U}\setminus \mathrm{W}_j}$, a subpacket index not previously used at any server is assigned.
\item For each message in ${\mathrm{U}\cap \mathrm{W}_j}$, a subpacket index not previously used at the side server is assigned; 
this index corresponds to a subpacket recovered in an earlier round from a target server.
\end{itemize}
\item In each target symbol, 
\begin{itemize}
\item For each message in ${\mathrm{U}\setminus \mathrm{W}_j}$, the same subpacket indices as in the side symbol are assigned.
\item For each message in ${\mathrm{U}\cap \mathrm{W}_j}$, indices corresponding to subpackets already recovered from the side server are assigned; 
these indices are shared across the target symbols and were not previously used at any target server.
\item The message $\mathrm{X}_i$ is assigned a subpacket index not previously used at any server, chosen distinctly across the target symbols.
\item For each message in ${\mathrm{V}\setminus \{i\}}$, a subpacket index not previously used at the target server is assigned; 
this index corresponds to a subpacket recovered in an earlier round from a different server.
\end{itemize}
\end{itemize}
\end{enumerate}

After subpacket indices have been assigned to all symbols involved in Rounds~$1$ through~$D$, any remaining symbols at each server, namely those not used for recovery, are indexed arbitrarily using indices not previously used at that server.

\subsection{Feasibility Constraints}

The variables $T_{\mathrm{U}}$, $I^{[\mathrm{W}_j]}_{\mathrm{U},\mathrm{V}}$, and $J^{[\mathrm{W}_j]}_{\mathrm{V},i}(k)$ represent symbol counts and are therefore non-negative integers, yielding  constraints~\eqref{eq:T_non_neg}, \eqref{eq:I_non_neg}, and~\eqref{eq:J_non_neg}.

Beyond integrality and non-negativity, these variables must satisfy additional constraints imposed by correctness and privacy for the proposed construction to yield a valid PSSR scheme.
In what follows, we derive these constraints.
   
Recall that for every support $\mathrm{U}\not\subseteq \mathrm{W}_j$, some symbols supported on $\mathrm{U}$ are designated as
\emph{side} symbols and some as \emph{target} symbols.
To ensure correctness, these role assignments must be feasible: at each server, the total number of symbols of support $\mathrm{U}$ used in either role cannot exceed the total number of retrieved symbols of that support, $T_{\mathrm{U}}$.

Fix a support set $\mathrm{U}\not\subseteq \mathrm{W}_j$.  
Per server, the number of symbols supported on $\mathrm{U}$ used as side for target symbols with supports ${\mathrm{U}\cup \mathrm{V}}$, over all non-empty $\mathrm{V}\subseteq \mathrm{W}_j\setminus \mathrm{U}$, equals 
\[\sum_{\mathrm{V} \in \mathbbm{V}_j(\mathrm{U})} I^{[\mathrm{W}_{j}]}_{\mathrm{U},\mathrm{V}},\] 
where $\mathbbm{V}_j(\mathrm{U}) = \{\mathrm{V}\subseteq \mathrm{W}_j\setminus \mathrm{U}: |\mathrm{V}| \neq 0\}$. 

Moreover, a side symbol from one server supported on $\mathrm{U}\setminus \mathrm{V}$, for any non-empty ${\mathrm{V}\subseteq \mathrm{W}_j\cap \mathrm{U}}$, can be paired with ${N-1}$ target symbols---one from each of the other ${N-1}$ servers---supported on $\mathrm{U}$.  
Accordingly, per server, the number of symbols supported on $\mathrm{U}$ used as target equals 
\[(N-1)\sum_{\mathrm{V} \in \mathbbm{V}_j([1:K]\setminus \mathrm{U})} I^{[\mathrm{W}_{j}]}_{\mathrm{U} \setminus \mathrm{V},\mathrm{V}},\] 
where $\mathbbm{V}_j([1:K]\setminus \mathrm{U}) = \{\mathrm{V}\subseteq \mathrm{W}_j\cap \mathrm{U}: |\mathrm{V}|\neq 0\}$. 

Since the user retrieves $T_{\mathrm{U}}$ symbols with support $\mathrm{U}$ from each server, the combined number of such symbols used as side and as target must not exceed $T_{\mathrm{U}}$. 
This yields constraint~\eqref{eq:set_1_cons}.

Correctness requires the user to recover exactly $L/N$ subpackets of each demand message $\mathrm{X}_i$, for every ${i\in\mathrm{W}_j}$.  
Per server, these subpackets come from: 
\begin{itemize}
\item Singleton symbols of  $\mathrm{X}_i$ retrieved directly, contributing $T_{\{i\}}$ subpackets; 
\item Single-demand symbols supported on $\{i\}$ generated via side--target pairings, where a target symbol supported on $\mathrm{U}\cup\{i\}$, where $\mathrm{U}\subseteq [1:K]\setminus \{i\}$ and $\mathrm{U}\not\subseteq\mathrm{W}_j$, is paired with a side symbol supported on $\mathrm{U}$, contributing the following number of subpackets:
\[(N-1)\sum_{\substack{\mathrm{U} \in \mathbbm{U}_j(\{i\})}}
I^{[\mathrm{W}_{j}]}_{\mathrm{U},\{i\}},\] 
where $\mathbbm{U}_j(\{i\}) = \{\mathrm{U}\subseteq [1:K]\setminus \{i\}: \mathrm{U}\not\subseteq \mathrm{W}_j\}$;  
\item Multi-demand symbols supported on ${\mathrm{V}\subseteq \mathrm{W}_j}$ with ${2\leq |\mathrm{V}|\leq k}$ and ${i\in\mathrm{V}}$ which are used in round $k$ (for ${k\in [2:D]}$) to recover new subpackets of $\mathrm{X}_i$, contributing the following number of subpackets: 
\[ \sum_{k=2}^{D} \sum_{l=1}^{k-1} \sum_{\substack{\mathrm{V}\in \mathbbm{V}_{j,l}(\{i\})}}    J^{[\mathrm{W}_{j}]}_{\mathrm{V}\cup \{i\},i} (k),\] 
where $\mathbbm{V}_{j,l}(\{i\}) = \{\mathrm{V}\subseteq \mathrm{W}_j\setminus \{i\}: |\mathrm{V}| = l\}$.  
\end{itemize} 
Equating the total contribution to $L/N$ yields  constraint~\eqref{eq:set_2_cons}.

Correctness also requires that, whenever the scheme uses demand-only symbols to recover new demand subpackets, there must be enough such symbols available at each server.
For any ${\mathrm{V}\subseteq \mathrm{W}_j}$, ${|\mathrm{V}|\geq 2}$, the demand-only symbols supported on $\mathrm{V}$ that are available at a given server consist of:
\begin{itemize}
\item the $T_{\mathrm{V}}$ symbols with support $\mathrm{V}$ retrieved directly; 
\item the symbols generated via side--target pairings, where a target symbol supported on $\mathrm{U}\cup \mathrm{V}$ with $\mathrm{U}\subseteq [1:K]\setminus \mathrm{V}$ and $\mathrm{U}\not\subseteq \mathrm{W}_j$ is paired with a side symbol supported on $\mathrm{U}$, contributing the following number of symbols with support $\mathrm{V}$: 
\[(N-1)\sum_{\substack{\mathrm{U} \in \mathbbm{U}_j(\mathrm{V})}}
I^{[\mathrm{W}_{j}]}_{\mathrm{U},\mathrm{V}},\] 
where $\mathbbm{U}_j(\mathrm{V}) = \{\mathrm{U}\subseteq [1:K]\setminus \mathrm{V}: \mathrm{U}\not\subseteq \mathrm{W}_j\}$.
\end{itemize}
These available symbols must cover their total usage across all rounds $k\in[|\mathrm{V}|:D]$ and all demand indices ${i\in \mathrm{V}}$, i.e., 
\[\sum_{k=|\mathrm{V}|}^{D} \sum_{i\in \mathrm{V}}J^{[\mathrm{W}_{j}]}_{\mathrm{V},i}(k).\]
This imposes  constraint~\eqref{eq:set_3_cons}.

For correctness, subpacket indices must be assigned so that, in each round and for each server, every demand subpacket used for cancellation has already been recovered from other servers in earlier rounds.
In particular, already recovered demand subpackets are needed for two types of cancellations:
\begin{itemize}
\item[(i)] canceling demand messages that appear in both the side and target symbols within a side--target pairing; 
\item[(ii)] canceling the remaining demand subpackets in a multi-demand symbol to recover a new demand subpacket.
\end{itemize}

Fix a demand index ${i\in \mathrm{W}_j}$. For any given server and any ${m\in[1:D-1]}$, the number of subpackets of the demand message $\mathrm{X}_i$ that are needed from the other servers by the end of round $m$---for use in round $m+1$ in the canceling operations in (i) and (ii)---cannot exceed the number of subpackets of $\mathrm{X}_i$ that are recovered from the other servers by the end of round $m$. 
This requirement is imposed by  constraint~\eqref{eq:indexing_cons_1}.

Specifically, the number of occurrences of the demand message $\mathrm{X}_i$ in both the side and target symbols within side--target pairings equals, per server and for each round ${k\in [2:D]}$, 
\begin{align}
    N \sum_{l=1}^{k-1} \sum_{\substack{\mathrm{V} \in \mathbbm{V}_{j,l}(\{i\})}} \sum_{\substack{\mathrm{U} \in \mathbbm{U}_{j,k-1-l}(\mathrm{V}\cup \{i\}) }}   I^{[\mathrm{W}_{j}]}_{\mathrm{U} \cup \{i\},\mathrm{V}}, \label{eq:indexing_cons_1_RHS_1} 
\end{align} 
where $\mathbbm{V}_{j,l}(\{i\})$ is defined as before, and  
\begin{align}
\mathbbm{U}_{j,k-1-l}(\mathrm{V}\cup \{i\})
&=
\bigl\{
\mathrm{U}\subseteq [1:K]\setminus (\mathrm{V}\cup \{i\}): \nonumber\\
&\qquad
\mathrm{U}\not\subseteq \mathrm{W}_j, \, 
|\mathrm{U}\cap \mathrm{W}_j|=k-1-l
\bigr\}.\nonumber
\end{align}
This expression counts (i) the number of occurrences of $\mathrm{X}_i$ in the \emph{side} symbols from a given server, and (ii) the corresponding number of \emph{target} symbols from that server, which in turn equals the number of occurrences of $\mathrm{X}_i$ in the \emph{side} symbols from the other ${N-1}$ servers paired with those target symbols. 
The sums range over all demand subsets ${\mathrm{V}\subseteq \mathrm{W}_j\setminus\{i\}}$ with ${|\mathrm{V}|=l}$ for ${l\in[1:k-1]}$, and over all subsets ${\mathrm{U}\subseteq [1:K]\setminus(\mathrm{V}\cup\{i\})}$ that include at least one interference message (i.e., ${\mathrm{U}\not\subseteq \mathrm{W}_j}$) and contribute exactly ${k-1-l}$ additional demand messages (i.e., ${|\mathrm{U}\cap \mathrm{W}_j|=k-1-l}$). These constraints ensure that the corresponding target support contains exactly $k$ demand messages---matching the round number
$k$---since ${|\mathrm{V}|+|\mathrm{U}\cap \mathrm{W}_j|+|\{i\}|=l+(k-1-l)+1=k}$.

In addition, the number of occurrences of the demand message $\mathrm{X}_i$ within multi-demand symbols used to recover new subpackets of other demand messages equals, per server and for each round ${k\in [2:D]}$,  
\begin{equation} \label{eq:indexing_cons_1_RHS_2}
    \sum_{l=1}^{k-1} \sum_{\substack{\mathrm{V} \in \mathbbm{V}_{j,l}(\{i\})}} \sum_{i' \in \mathrm{V} }  J^{[\mathrm{W}_{j}]}_{\mathrm{V} \cup \{i\},i'} (k),
\end{equation}
where $\mathbbm{V}_{j,l}(\{i\})$ is defined as before.
Here, the sums range over all demand subsets ${\mathrm{V}\subseteq \mathrm{W}_j\setminus \{i\}}$ with ${|\mathrm{V}| = l}$ for ${l\in [1:k-1]}$, and over all demand indices ${i'\in \mathrm{V}}$. 
These constraints ensure that the corresponding multi-demand symbols in round $k$ involve at most $k$ demand messages,
since ${|\mathrm{V}|+|\{i\}|=l+1\leq k}$.

Combining these, for any given server and any ${m\in [1:D-1]}$, the number of subpackets of the demand message $\mathrm{X}_i$ that must be available from the other servers by the end of round $m$ equals the sum of~\eqref{eq:indexing_cons_1_RHS_1} and~\eqref{eq:indexing_cons_1_RHS_2} over  $k\in [2:m+1]$, which matches the RHS of  constraint~\eqref{eq:indexing_cons_1}. 
Moreover, for any given server and any ${m\in [1:D-1]}$, the subpackets of the demand message $\mathrm{X}_i$ recovered from the other servers by the end of round $m$ come from: 
\begin{itemize}
\item the singleton symbols of  $\mathrm{X}_i$ retrieved directly, contributing ${(N-1)T_{\{i\}}}$ subpackets; 
\item the single-demand symbols supported on $\{i\}$ generated via side--target pairings, where a target symbol supported on ${\mathrm{U}\cup\{i\}}$, where ${\mathrm{U}\subseteq [1:K]\setminus \{i\}}$, ${\mathrm{U}\not\subseteq\mathrm{W}_j}$, and ${|\mathrm{U}\cap\mathrm{W}_j|\leq m-1}$, is paired with a side symbol supported on $\mathrm{U}$, contributing the following number of subpackets: 
\[(N-1)^2 \sum_{k=1}^{m}\ \sum_{\substack{\mathrm{U} \in \mathbbm{U}_{j,k-1}(\{i\})}}
I^{[\mathrm{W}_{j}]}_{\mathrm{U},\{i\}},\] 
where 
\begin{align}
\mathbbm{U}_{j,k-1}(\{i\})
&=
\bigl\{
\mathrm{U}\subseteq [1:K]\setminus \{i\}: \nonumber\\
&\qquad
\mathrm{U}\not\subseteq \mathrm{W}_j, \, 
|\mathrm{U}\cap \mathrm{W}_j| = k-1
\bigr\};\nonumber
\end{align}  
\item the multi-demand symbols supported on ${\mathrm{V}\subseteq \mathrm{W}_j}$ with ${2\leq |\mathrm{V}|\leq k}$ and ${i\in\mathrm{V}}$ which are used in round $k$ (for ${k\in [2:m]}$) to recover new subpackets of $\mathrm{X}_i$, contributing the following number of subpackets: 
\[(N-1)\sum_{k=2}^{m} \sum_{l=1}^{k-1} \sum_{\substack{\mathrm{V} \in \mathbbm{V}_{j,l}(\{i\})}}  J^{[\mathrm{W}_{j}]}_{\mathrm{V} \cup \{i\},i} (k),\] 
where $\mathbbm{V}_{j,l}(\{i\})$ is defined as before.  
\end{itemize}
Summing these three contributions yields the LHS of  constraint~\eqref{eq:indexing_cons_1}. 

Finally, correctness also requires per-server non-reuse of subpacket indices: 
among the symbols the user retrieves from any given server, a fixed message subpacket may not appear more than once.
Equivalently, for each message $\mathrm{X}_i$, the total number of symbols in which $\mathrm{X}_i$ appears at a given server,  
\[\sum_{\substack{\mathrm{U}\subseteq[1:K]\\ i\in \mathrm{U}}} T_{\mathrm{U}},\] 
cannot exceed the number of subpackets per message, $L$, yielding constraint~\eqref{eq:indexing_cons_2}.

With these feasibility constraints in place, privacy follows by construction and imposes no additional constraints.
Fix any server.
For every support $\mathrm{U}$, the user's query specifies the total number of symbols of that support;
these per-support counts are given by $\{T_{\mathrm{U}}\}$ and are identical for all candidate demand index sets.
Therefore, at this server, the number of occurrences of each message across all symbols is fully 
determined by $\{T_{\mathrm{U}}\}$ and is independent of the demand index set.

Given these supports, subpacket indices are assigned subject to per-server non-reuse constraints, so that each message contributes,
at that server, a fixed number of \emph{distinct} subpackets determined by $\{T_{\mathrm{U}}\}$ and hence the same for all candidate demand index sets.
Since each message is randomly partitioned into subpackets, the resulting indices appear as uniformly random labels to
the server and therefore reveal no information about the demand index set.
Thus, the scheme satisfies privacy.

\subsection{Maximizing the Achievable Rate}

For any $L$ divisible by $N$, every assignment of the variables $T_{\mathrm{U}}$ and the auxiliary variables $I^{[\mathrm{W}_j]}_{\mathrm{U},\mathrm{V}}$ and $J^{[\mathrm{W}_j]}_{\mathrm{V},i}(k)$ that satisfies the constraints~\eqref{eq:set_1_cons}--\eqref{eq:J_non_neg} specifies a valid PSSR scheme.

A natural question is whether these constraints admit any feasible assignment.
This is indeed the case, and feasibility is immediate: 
set ${T_{\{i\}}=L/N}$ for all ${i\in[1:K]}$,
set ${T_{\mathrm{U}}=0}$ for all ${\mathrm{U}\subseteq[1:K]}$ with ${|\mathrm{U}|\geq 2}$,
and set every auxiliary variable
$I^{[\mathrm{W}_j]}_{\mathrm{U},\mathrm{V}}$ and
$J^{[\mathrm{W}_j]}_{\mathrm{V},i}(k)$ to zero. 
It is straightforward to verify that this assignment satisfies the constraints
\eqref{eq:set_1_cons}--\eqref{eq:J_non_neg} and corresponds to the naive scheme in which the user retrieves from each server exactly $L/N$ distinct subpackets from every message. 
While this scheme is feasible for any arbitrary demand family, it may not be rate-optimal, which motivates the optimization framework presented next.

For any feasible assignment, the user retrieves from each server a total of 
\[
T \coloneqq \sum_{\substack{\mathrm{U}\subseteq [1:K]}} T_{\mathrm{U}}
\]
$\mathbb{F}_q$-symbols. 
Across all $N$ servers the user retrieves $NT$ $\mathbb{F}_q$-symbols to recover $DL$ demand subpackets, each an
$\mathbb{F}_q$-symbol. 
Thus, the retrieval rate, as defined in~\eqref{eq:RateDef}, equals $DL/(NT)$. 

For a fixed $L$, maximizing the retrieval rate
is equivalent to minimizing $T$ subject to~\eqref{eq:set_1_cons}--\eqref{eq:J_non_neg}, which is an integer linear program (ILP).
Allowing $L$ to vary, we can therefore maximize the achievable rate over all $L$; 
equivalently, we can minimize the ratio $T/L$, subject to \eqref{eq:set_1_cons}--\eqref{eq:J_non_neg} and the additional
integrality and positivity constraint on $L$ in~\eqref{eq:L_geq_1}. 

While the constraints are linear in the variables $T_{\mathrm{U}}$, $I^{[\mathrm{W}_j]}_{\mathrm{U},\mathrm{V}}$, $J^{[\mathrm{W}_j]}_{\mathrm{V},i}(k)$, and $L$, the objective is fractional.
Normalizing all variables by $L$ yields an equivalent formulation in which both the objective and the constraints are linear in the normalized variables.
Consequently, the problem admits a linear programming (LP) reformulation.
Solving this LP provides an optimal assignment of the normalized variables; we then choose $L$ as the smallest positive integer that makes the corresponding unnormalized variables integral.

\subsection{Minimizing the Subpacketization Level}

The LP above characterizes the maximum achievable rate within the underlying
class of schemes. 
However, this rate can in general be attained with different subpacketization levels, and the LP realization need not use the smallest possible $L$.

To identify a maximum-rate scheme with minimum $L$, we introduce an additional
constraint, as in~\eqref{eq:min_L_cons}, which requires the scheme's retrieval rate
$DL/(NT)$ to match the LP-optimal value, and then minimize $L$ subject to~\eqref{eq:set_1_cons}--\eqref{eq:L_geq_1} and~\eqref{eq:min_L_cons}.
Since $L$ and all counting variables are integer-valued, the resulting formulation is an ILP.

\section{An Illustrative Example}

In this section, we present an illustrative example of the proposed PSSR scheme.  

Suppose there are $N=2$ servers, each storing an identical copy of
$K=5$ messages. 
We index the messages by $1,2,3,4,5$ and, for notational convenience, denote them by $a,b,c,d,e$, respectively.

Suppose a user wishes to retrieve ${D=2}$ messages, indexed by one of the following ${E=4}$ candidate demand index sets:
\[\mathrm{W}_1 = \{1,3\},\; \mathrm{W}_2 = \{2,3\},\; \mathrm{W}_3 = \{3,4\},\; \mathrm{W}_4 = \{4,5\}.\]

For this example, solving the optimization problem in Theorem~\ref{thm:PRSR capacity} yields the achievable rate ${R_{*} = 8/13}$. 
With $R_{*}$ fixed, the optimization problem in Theorem~\ref{thm:PRSR_L*_bounds} yields the achievable subpacketization level ${L^{*}=8}$ and the following values of $T_{\mathrm{U}}$ for all non-empty subsets ${\mathrm{U}\subseteq[1:5]}$:
\[
T_{\{2\}}=1,\; T_{\{3\}}=2,\; T_{\{4\}}=1,\; T_{\{5\}}=2,
\]
\[
T_{\{1,3\}}=1,\; T_{\{2,4\}}=1,\; T_{\{3,5\}}=2,
\]
\[
T_{\{1,2,3\}}=1,\; T_{\{1,3,4\}}=1,
\]
\[
T_{\{1,2,3,4\}}=1,
\]
and $T_{\mathrm{U}}=0$ for all other $\mathrm{U}$.
This means that, for any $\mathrm{W}_j$ with $j \in [1:4]$, the user retrieves from each server one singleton symbol from each of $b$ and $d$, two singleton symbols from each of $c$ and $e$, one symbol (linear combination) involving subpackets from $a$ and $c$, one involving subpackets from $b$ and $d$, one involving subpackets from $a$, $b$, and $c$, one involving subpackets from $a$, $c$, and $d$, one involving subpackets from $a$, $b$, $c$, and $d$, and two symbols involving subpackets from $c$ and $e$.

Suppose each message is randomly and independently divided into ${L^{*}=8}$ subpackets, denoted by $a_1, \dots, a_8$, $b_1, \dots, b_8$, $c_1, \dots, c_8$, $d_1, \dots, d_8$, and $e_1, \dots, e_8$. 

{\renewcommand{\arraystretch}{1.105}
\begin{table}[t]
    \centering
    \caption{Query table for the case  $\mathrm{W}=\mathrm{W}_1 = \{1,3\}$\\ (demand messages: $a$ and $c$) }\label{tab:PSSR_W1}
    \scalebox{1.125}{
    \begin{tabular}{|c|c|}
    \hline
     Server 1 & Server 2\\
    \hline
    \hline
     $b_{1}$ & $b_{2}$ \\
      $c_{1}$, $c_{2}$ & $c_{3}$, $c_{4}$ \\
     $d_{1}$  & $d_{2}$ \\
     $e_{1}$, $e_{2}$ & $e_{3}$, $e_{4}$ \\
    \hline
    \hline
     $a_{1}+c_{3}$ & $a_{2}+c_{1}$\\
     $b_{3}+d_{3}$ & $b_{4}+d_{4}$ \\
     $c_{5}+e_{3}$, $c_{6}+e_{4}$ & $c_{7}+e_{1}$, $c_{8}+e_{2}$ \\
    \hline
    \hline
     $a_{3}+b_{2}+c_{4}$ & $a_{4}+b_{1}+c_{2}$ \\
     $a_{5}+c_{7}+d_{2}$ & $a_{6}+c_{5}+d_{1}$\\
    \hline
    \hline
     $a_{7}+b_{4}+c_{8}+d_{4}$ & $a_{8}+b_{3}+c_{6}+d_{3}$\\
    \hline
    \end{tabular}
    }
\end{table}
}

{\renewcommand{\arraystretch}{1.105}
\begin{table}[t]
    \centering
    \caption{Query table for the case  $\mathrm{W}=\mathrm{W}_2 = \{2,3\}$\\ (demand messages: $b$ and $c$) }\label{tab:PSSR_W2}
    \scalebox{1.125}{
    \begin{tabular}{|c|c|}
    \hline
     Server 1 & Server 2\\
    \hline
    \hline
     $b_{1}$ & $b_{2}$ \\
      $c_{1}$, $c_{2}$ & $c_{3}$, $c_{4}$ \\
     $d_{1}$  & $d_{2}$ \\
     $e_{1}$, $e_{2}$ & $e_{3}$, $e_{4}$  \\
    \hline
    \hline
     $a_{1}+c_{3}$ & $a_{2}+c_{1}$\\
     $b_{3}+d_{2}$ & $b_{4}+d_{1}$ \\
     $c_{5}+e_{3}$, $c_{6}+e_{4}$ & $c_{7}+e_{1}$, $c_{8}+e_{2}$ \\
     \hline
     \hline
     $a_{2}+b_{5}+c_{4}$ & $a_{1}+b_{6}+c_{2}$ \\
     $a_{3}+c_{7}+d_{3}$ & $a_{4}+c_{5}+d_{4}$\\
    \hline
    \hline
     $a_{4}+b_{7}+c_{8}+d_{4}$ & $a_{3}+b_{8}+c_{6}+d_{3}$\\
    \hline
    \end{tabular}
    }
\end{table}
}

{\renewcommand{\arraystretch}{1.105}
\begin{table}[t]
    \centering
    \caption{Query table for the case  $\mathrm{W}=\mathrm{W}_3 = \{3,4\}$\\ (demand messages: $c$ and $d$) }\label{tab:PSSR_W3}
    \scalebox{1.125}{
    \begin{tabular}{|c|c|}
    \hline
     Server 1 & Server 2\\
    \hline
    \hline
     $b_{1}$ & $b_{2}$ \\
      $c_{1}$, $c_{2}$ & $c_{3}$, $c_{4}$ \\
     $d_{1}$  & $d_{2}$ \\
     $e_{1}$, $e_{2}$ & $e_{3}$, $e_{4}$ \\
    \hline
    \hline
     $a_{1}+c_{3}$ & $a_{2}+c_{1}$\\
     $b_{2}+d_{3}$ & $b_{1}+d_{4}$ \\
     $c_{5}+e_{3}$, $c_{6}+e_{4}$ & $c_{7}+e_{1}$, $c_{8}+e_{2}$ \\
    \hline
    \hline
     $a_{3}+b_{3}+c_{7}$ & $a_{4}+b_{4}+c_{5}$ \\
     $a_{2}+c_{4}+d_{5}$ & $a_{1}+c_{2}+d_{6}$\\
    \hline
    \hline
     $a_{4}+b_{4}+c_{8}+d_{7}$ & $a_{3}+b_{3}+c_{6}+d_{8}$\\
    \hline
    \end{tabular}
    }
\end{table}
}

{\renewcommand{\arraystretch}{1.105}
\begin{table}[t]
    \centering
    \caption{Query table for the case  $\mathrm{W}=\mathrm{W}_4=\{4,5\}$\\ (demand messages: $d$ and $e$) }\label{tab:PSSR_W4}
    \scalebox{1.125}{
    \begin{tabular}{|c|c|}
    \hline
     Server 1 & Server 2\\
    \hline
    \hline
     $b_{1}$ & $b_{2}$ \\
      $c_{1}$, $c_{2}$ & $c_{3}$, $c_{4}$ \\
     $d_{1}$  & $d_{2}$ \\
     $e_{1}$, $e_{2}$ & $e_{3}$, $e_{4}$ \\
    \hline
    \hline
     $a_{1}+c_{5}$ & $a_{2}+c_{6}$\\
     $b_{2}+d_{3}$ & $b_{1}+d_{4}$ \\
     $c_{3}+e_{5}$, $c_{4}+e_{6}$ & $c_{1}+e_{7}$, $c_{2}+e_{8}$ \\
    \hline
    \hline
     $a_{3}+b_{3}+c_{7}$ & $a_{4}+b_{4}+c_{8}$ \\
     $a_{2}+c_{6}+d_{5}$ & $a_{1}+c_{5}+d_{6}$\\
    \hline
    \hline
     $a_{4}+b_{4}+c_{8}+d_{7}$ & $a_{3}+b_{3}+c_{7}+d_{8}$\\
    \hline
    \end{tabular}
    }
\end{table}
}

Table~\ref{tab:PSSR_W1} lists the query sent to each server when the user's demand index set is $\mathrm{W}_1$, that is, when the demand messages are $a$ and $c$. 
Similarly, Tables~\ref{tab:PSSR_W2},~\ref{tab:PSSR_W3}, and~\ref{tab:PSSR_W4} list the queries corresponding to the demand index sets $\mathrm{W}_2$, $\mathrm{W}_3$, and $\mathrm{W}_4$, respectively. 
In each table, the subpacket indices are assigned according to the procedure described in Section~\ref{sec: Subpacket Indexing}. 
Below, Tables~\ref{tab:PSSR_W1} and~\ref{tab:PSSR_W3} are discussed in detail, with the subpacket indexing explained together with the corresponding recovery process to establish correctness. 
The proofs for Tables~\ref{tab:PSSR_W2} and~\ref{tab:PSSR_W4} follow in the same way.

\subsection{Proof of Correctness}

Consider Table~\ref{tab:PSSR_W1}, which corresponds to the demand index set $\mathrm{W}_1 = \{1,3\}$, i.e., the demand messages are $a$ and $c$. 
An optimal solution to the optimization problem yields the following values of $I^{[\mathrm W_1]}_{\mathrm U,\mathrm V}$ for all pairs ${(\mathrm{U},\mathrm{V})}$ with ${\mathrm{U}\subseteq[1:5]}$, ${\mathrm{U}\not\subseteq \mathrm{W}_1}$, and ${\mathrm{V}\subseteq \mathrm{W}_1\setminus \mathrm{U}}$: 
\[
I^{[\mathrm{W}_{1}]}_{\{5\},\{3\}}=2\;,\;
I^{[\mathrm{W}_{1}]}_{\{2\},\{1,3\}}=1\;,\;
I^{[\mathrm{W}_{1}]}_{\{4\},\{1,3\}}=1\;,\;
I^{[\mathrm{W}_{1}]}_{\{2,4\},\{1,3\}}=1,
\] 
and $I^{[\mathrm W_1]}_{\mathrm{U},\mathrm{V}}=0$ for all other pairs ${(\mathrm{U},\mathrm{V})}$. 
Additionally, it yields the following values of $J^{[\mathrm W_1]}_{\mathrm V,i}(k)$ for all pairs ${(\mathrm{V},i)}$ and $k$ with ${\mathrm{V} \subseteq \mathrm W_1}$, ${|\mathrm{V}|\geq 2}$, $i \in \mathrm{V}$, and ${k \in [|\mathrm{V}|:2]}$, which here reduces to ${(\mathrm{V},i)\in \{(\{1,3\},1),(\{1,3\},3)\}}$ and ${k=2}$:
\[
J^{[\mathrm{W}_{1}]}_{\{1,3\},1} (2)=4\;,\; J^{[\mathrm{W}_{1}]}_{\{1,3\},3} (2)=0.
\]

Subpacket indexing and the recovery process take place over ${D=2}$ rounds. 
We begin with Round~1.
Assigning distinct subpacket indices to the retrieved singleton symbols of each message from each server, the user retrieves $b_1$, $c_1$, $c_2$, $d_1$, $e_1$, and $e_2$ from Server~1 and $b_2$, $c_3$, $c_4$, $d_2$, $e_3$, and $e_4$ from Server~2.
Thus, the user directly recovers $c_1$, $c_2$, $c_3$, and $c_4$. 

Next, consider side--target pairings in which the target contains exactly one demand message. 
Since $I^{[\mathrm{W}_{1}]}_{\{5\},\{3\}}=2$, where $\{5\}$ corresponds to the message $e$ and $\{3\}$ corresponds to the message $c$, the two side symbols $e_3$ and $e_4$ retrieved from Server~2 must be combined with the two target symbols from Server~1 involving the messages $c$ and $e$.
Thus, the two subpackets of $e$ appearing in these two symbols from Server~1 are assigned the indices $3$ and $4$, while two not-previously-used indices, $5$ and $6$, are assigned to the subpackets of the demand message $c$. 
By combining the two side symbols $e_3$ and $e_4$ from Server~2 with their corresponding target symbols from Server~1, the user recovers $c_5$ and $c_6$.
Likewise, the two target symbols from Server~2 involving the messages $c$ and $e$ are indexed in the same way, yielding the symbols $c_7+e_1$ and $c_8+e_2$, from which the user recovers $c_7$ and $c_8$.
This completes the subpacket indexing and the recovery process in Round~1.

Next, we describe the subpacket indexing and the recovery process in Round~2. 
This round consists of (i) demand-only symbols that are retrieved directly and contain both $a$ and $c$, and (ii) side--target pairings whose targets involve both $a$ and $c$.
Since ${J^{[\mathrm{W}_{1}]}_{\{1,3\},1} (2)=4}$, there are four demand-only symbols (either directly retrieved or obtained by subtracting a side symbol from its corresponding target symbol) in Round~2 that enable recovery of four subpackets of $a$ from each server.

Start with the directly retrieved demand-only symbol from Server~1 involving $a$ and $c$. 
Since this symbol must enable the recovery of a subpacket of $a$, the message $a$ is assigned index $1$, while the message $c$ is assigned index $3$ since $c_3$ was recovered from Server~2 in Round~1.
Likewise, in the corresponding demand-only symbol from Server~2, $a$ and $c$ are assigned indices $2$ and $1$, respectively.
Thus, the user recovers $a_1$ and $a_2$.

We next consider side--target pairings whose targets involve both $a$ and $c$.
Since ${I^{[\mathrm{W}_{1}]}_{\{2\},\{1,3\}}=1}$, the side symbol $b_2$ retrieved from Server~2 is paired with the target symbol from Server~1 involving $a$, $b$, and $c$. 
Thus, in this target symbol, $b$ is assigned index $2$, $c$ is assigned index $4$ since $c_4$ was previously recovered from Server~2, and $a$ is assigned index $3$ since $a_3$ was not previously used at any server.
Similarly, in the corresponding target symbol retrieved from Server~2, $a$, $b$, and $c$ are assigned indices $4$, $1$, and $2$, respectively. 
Thus, the user recovers $a_3$ and $a_4$.

Since ${I^{[\mathrm{W}_{1}]}_{\{4\},\{1,3\}}=1}$, the side symbol $d_2$ from Server~2 is paired with the target symbol from Server~1 involving $a$, $c$, and $d$.
Thus, in this target symbol, $d$ is assigned index $2$, $c$ is assigned index $7$ since $c_7$ was previously recovered from Server~2, and $a$ is assigned index $5$ since $a_5$ was not used previously at any server. 
Similarly, in the corresponding target symbol from Server~2, $a$, $c$, and $d$ are assigned indices $6$, $5$, and $1$, respectively.
Thus, the user recovers $a_5$ and $a_6$.

Finally, since ${I^{[\mathrm{W}_{1}]}_{\{2,4\},\{1,3\}}=1}$, the symbol involving $b$ and $d$ is used as side information. 
Since both $b$ and $d$ are interference messages, not-previously-used indices must be assigned to them. 
Accordingly, in Server~1, both $b$ and $d$ are assigned index $3$, while in Server~2, both are assigned index $4$.
Now consider the corresponding target symbol from Server~1 involving $a$, $b$, $c$, and $d$. 
The messages $b$ and $d$ are assigned the same indices as in the side symbol from Server~2, namely, index $4$ for both $b$ and $d$.
The message $c$ is assigned index $8$ since $c_8$ was previously recovered from Server~2, while the message $a$ is assigned index $7$ since $a_7$ was not used previously at any server. 
Likewise, in the corresponding target symbol from Server~2, $a$, $b$, $c$, and $d$ are assigned indices $8$, $3$, $6$, and $3$, respectively.
Thus, the user recovers $a_7$ and $a_8$.
This completes the subpacket indexing and the recovery process in Round~2 and establishes correctness for the case ${\mathrm{W} = \mathrm{W}_1}$. 

Next, consider Table~\ref{tab:PSSR_W3}, which corresponds to the demand index set ${\mathrm{W}_3=\{3,4\}}$, i.e., the demand messages are $c$ and $d$. 
An optimal solution to the optimization problem yields the following values of $I^{[\mathrm W_3]}_{\mathrm U,\mathrm V}$ for all pairs ${(\mathrm{U},\mathrm{V})}$ with ${\mathrm{U}\subseteq[1:5]}$, ${\mathrm{U}\not\subseteq \mathrm{W}_3}$, and ${\mathrm{V}\subseteq \mathrm{W}_3\setminus \mathrm{U}}$: 
\[
I^{[\mathrm{W}_{3}]}_{\{5\},\{3\}}=2\;,\;
I^{[\mathrm{W}_{3}]}_{\{2\},\{4\}}=1\;,\;
I^{[\mathrm{W}_{3}]}_{\{1,3\},\{4\}}=1\;,\;
I^{[\mathrm{W}_{3}]}_{\{1,2,3\},\{4\}}=1,
\] 
and $I^{[\mathrm W_3]}_{\mathrm{U},\mathrm{V}}=0$ for all other pairs ${(\mathrm{U},\mathrm{V})}$. 
Additionally, it yields the following values of $J^{[\mathrm W_3]}_{\mathrm V,i}(k)$ for all pairs ${(\mathrm{V},i)}$ and $k$ with ${\mathrm{V} \subseteq \mathrm W_3}$, ${|\mathrm{V}|\geq 2}$, $i \in \mathrm{V}$, and ${k \in [|\mathrm{V}|:2]}$, which here reduces to ${(\mathrm{V},i)\in \{(\{3,4\},3),(\{3,4\},4)\}}$ and ${k=2}$:
\[
J^{[\mathrm{W}_{3}]}_{\{3,4\},3} (2)=0\;,\; J^{[\mathrm{W}_{3}]}_{\{3,4\},4} (2)=0.
\]

Starting with Round~1, subpacket indexing for the singleton symbols and for the symbols involving $c$ and $e$ proceeds in the same way as for ${\mathrm{W}_1=\{1,3\}}$, since ${I^{[\mathrm{W}_{3}]}_{\{5\},\{3\}}=I^{[\mathrm{W}_{1}]}_{\{5\},\{3\}}=2}$. 
Thus, the user recovers all the subpackets of $c$ in addition to $d_1$ and $d_2$.
Additionally, since ${I^{[\mathrm{W}_{3}]}_{\{2\},\{4\}}=1}$, the singletons $b_1$ and $b_2$ are used as side symbols to recover new subpackets of $d$ from the target symbols involving $b$ and $d$. 
Accordingly, in the corresponding target symbol from Server~1, $b$ and $d$ are assigned indices $2$ and $3$, respectively, while in the corresponding target symbol from Server~2, they are assigned indices $1$ and $4$, respectively.
Thus, the user recovers $d_3$ and $d_4$. 
This completes the subpacket indexing and the recovery process in Round~1. 

Next, we describe the subpacket indexing and the recovery process in Round~2. 
Since ${I^{[\mathrm{W}_{3}]}_{\{1,3\},\{4\}}=1}$, the side symbols involving $a$ and $c$ are paired with the target symbols involving $a$, $c$, and $d$. Accordingly, in the corresponding side symbol from Server~1, $a$ is assigned index $1$, while $c$ is assigned index $3$ since $c_3$ was previously recovered from Server~2. 
Likewise, in the corresponding side symbol from Server~2, $a$ and $c$ are assigned indices $2$ and $1$, respectively. 
Additionally, in the corresponding target symbol from Server~1, $a$ is assigned the same index $2$ as in the corresponding side symbol from Server~2, $c$ is assigned index $4$ since $c_4$ was previously recovered from Server~2, and $d$ is assigned index $5$ since $d_5$ was not used previously at any server. 
Similarly, in the corresponding target symbol from Server~2, $a$, $c$, and $d$ are assigned indices $1$, $2$, and $6$, respectively.
By combining these side and target symbols, the user recovers $d_5$ and $d_6$. 

Finally, since ${I^{[\mathrm{W}_{3}]}_{\{1,2,3\},\{4\}}=1}$, the side symbols involving $a$, $b$, and $c$ are paired with the target symbols involving $a$, $b$, $c$, and $d$.
Accordingly, in the corresponding side symbol from Server~1, both $a$ and $b$ are assigned index $3$ since $a_3$ and $b_3$ were not used previously at any server, and $c$ is assigned index $7$ since $c_7$ was previously recovered from Server~2.
Similarly, in the corresponding side symbol from Server~2, $a$, $b$, and $c$ are assigned indices $4$, $4$, and $5$, respectively. 
Additionally, in the corresponding target symbol from Server~1, both $a$ and $b$ are assigned the same index $4$ as in the corresponding side symbol from Server~2, $c$ is assigned index $8$ since $c_8$ was previously recovered from Server~2, and $d$ is assigned index $7$ since $d_7$ was not used previously at any server.
Similarly, in the corresponding target symbol from Server~2, $a$, $b$, $c$, and $d$ are assigned indices $3$, $3$, $6$, and $8$, respectively.
By combining these side and target symbols, the user recovers $d_7$ and $d_8$.
This completes the subpacket indexing and the recovery process in Round~2 and establishes correctness for the case ${\mathrm{W} = \mathrm{W}_3}$. 

\subsection{Proof of Privacy}

Privacy follows from the fact that, for any fixed server, the query has the same distribution for every candidate demand index set. 
Indeed, across Tables~\ref{tab:PSSR_W1}--\ref{tab:PSSR_W4}, the query seen by each server is structurally identical: 
the same supports appear with the same multiplicities in all cases, each message appears in the same pattern of retrieved symbols, and the subpackets of every message appearing at that server are distinct. 
The only variation across the tables lies in the particular subpacket indices appearing in those symbols. 
Since each message is randomly partitioned into subpackets, these indices appear to the server as uniformly random labels. 
Therefore, the query observed by any server is independent of the user's demand index set, which establishes privacy.

\subsection{Proof of Optimality}

In this example, the user retrieves $13$ symbols from each server, where each symbol has the size of one message subpacket, in order to recover all ${L^{*}=8}$ subpackets of the ${D=2}$ demand messages.
Thus, the retrieval rate of the proposed scheme is
\[
R_{*}=\frac{2\cdot 8}{2\cdot 13}=\frac{8}{13}.
\]
This matches the rate upper bound $R^{*}$ in Theorem~\ref{thm:PRSR capacity}, which, for the permutation  
\[\pi = \begin{pmatrix} 1 & 2 & 3 & 4\\ 1 & 4 & 2 & 3\end{pmatrix},\]
evaluates to
\begin{equation*}
R^{*} = 2\left(2+\frac{1}{2}\cdot 2+\frac{1}{4}\cdot 1+\frac{1}{8}\cdot 0 \right)^{-1} = \frac{8}{13}.
\end{equation*}
Thus, the proposed scheme is rate-optimal for this example.

Moreover, for the achievable rate ${R_{*} = 8/13}$, Theorem~\ref{thm:PRSR_L*_bounds} yields the lower bound 
\[
L_{*}= \frac{2\cdot 8}{\gcd(2 \cdot 8, 2 \cdot 13)}=8
\] on the subpacketization level. 
Since the proposed scheme uses subpacketization level ${L^{*}=8}$, it is therefore also optimal with respect to subpacketization level in this example. 

\subsection{Comparison with the MPIR Setting}

For comparison, consider the MPIR setting with ${N=2}$, ${K=5}$, and ${D=2}$, in which the candidate demand family consists of all pairs of messages. 
Since the PSSR demand family in this example is a subset of the MPIR demand family, any MPIR scheme for this parameter setting is also applicable to the PSSR instance considered here. 
The best-known MPIR scheme for this setting, due to~\cite{HWS2025}, achieves rate $82/135$, which is strictly smaller than the rate achieved by the proposed PSSR scheme, namely ${R_{*}=8/13}$, and requires subpacketization level $82$, which is substantially larger than the subpacketization level of the proposed scheme, namely ${L^{*}=8}$. 
Therefore, in this example, exploiting the restricted demand structure in PSSR yields gains in both achievable rate and subpacketization level relative to the MPIR setting.

\section{Computational Complexity and Symmetry Reduction}

In this section, we analyze the computational complexity of evaluating the converse bound and the achievable rate, and discuss how automorphisms of the demand family can reduce the complexity of both computations.

\subsection{Computation of the Converse Bound}

We analyze the computational complexity of brute-force evaluation of the converse bound $R^{*}$ in~\eqref{eq:PSR ub}.

Fix an arbitrary demand family with $E$ candidate demands, each of size $D$.
A direct brute-force implementation enumerates all permutations ${\pi:[1:E]\to[1:E]}$.
For each permutation, the objective can be computed using set unions and set differences over at most $E$ sets of size $D$, with cost $O(ED)$.
Thus, the direct brute-force complexity is $O(E!ED)$.

For the full demand family, $E=\binom{K}{D}$.
Without exploiting automorphisms, brute-force evaluation checks all $E!$ orderings and has complexity $O(E!ED)$.
Relabeling the $K$ message indices preserves the full demand family and leaves the objective unchanged.
These relabelings form an automorphism group of size $K!$.
Since no non-identity relabeling can fix an ordering of all $D$-subsets, then each orbit contains exactly $K!$ equivalent orderings, and the number of orbit representatives is $E!/K!$.
The resulting brute-force complexity is $O((E!/K!)ED)$.

\subsection{Computation of the Achievable Rate}

We count the variables and constraints in the optimization formulation used to evaluate the achievable rate $R_{*}$ in~\eqref{eq:PSSR lb}.

Fix an arbitrary demand family. 
We begin with the variables. 
The variables $T_{\mathrm{U}}$ are indexed by all non-empty subsets ${\mathrm{U}\subseteq [1:K]}$, and hence their number is ${2^K-1}$, which grows exponentially in $K$.
The variables $I^{[\mathrm{W}_{j}]}_{\mathrm{U},\mathrm{V}}$ are indexed by ${j\in[1:E]}$, ${\mathrm{V}\subseteq \mathrm{W}_j}$, $\mathrm{V}\neq \emptyset$, and ${\mathrm{U}\subseteq [1:K]\setminus \mathrm{V}}$, ${\mathrm{U}\not\subseteq \mathrm{W}_j}$. 
Fix ${j\in[1:E]}$ and a non-empty subset $\mathrm{V}\subseteq \mathrm{W}_j$. 
Then the number of admissible subsets $\mathrm{U}$ is ${2^{D-|\mathrm{V}|}(2^{K-D}-1)}$. 
Summing over all ${j\in[1:E]}$ and all non-empty ${\mathrm{V}\subseteq \mathrm{W}_j}$, the total number of variables $I^{[\mathrm{W}_{j}]}_{\mathrm{U},\mathrm{V}}$ is ${E(2^{K-D}-1)(3^D-2^D)}$, which grows linearly in $E$ and exponentially in both $K-D$ and $D$.
The variables $J^{[\mathrm{W}_{j}]}_{\mathrm{V},i}(k)$ are indexed by ${j\in[1:E]}$, ${\mathrm{V}\subseteq \mathrm{W}_j}$ with ${|\mathrm{V}|\geq 2}$, ${i\in \mathrm{V}}$, and ${k\in[|\mathrm{V}|:D]}$. 
For fixed $j$ and $\mathrm{V}$, the number of pairs ${(i,k)}$ is ${|\mathrm{V}|(D-|\mathrm{V}|+1)}$. 
Summing over all ${j\in[1:E]}$ and all ${\mathrm{V}\subseteq \mathrm{W}_j}$ with ${|\mathrm{V}|\geq 2}$, the total number of variables $J^{[\mathrm{W}_{j}]}_{\mathrm{V},i}(k)$ is ${ED((D+1)2^{D-2}-D)}$, which grows linearly in $E$ and exponentially in $D$.

We next count the constraints. 
The condition in~\eqref{eq:set_1_cons} yields ${E(2^K-2^D)}$ constraints, which grows exponentially in $K$. 
The condition in~\eqref{eq:set_2_cons} yields ${ED}$ constraints, which grows linearly in both $E$ and $D$. 
The condition in~\eqref{eq:set_3_cons} yields ${E(2^D-D-1)}$ constraints, which grows linearly in $E$ and exponentially in $D$. 
The condition in~\eqref{eq:indexing_cons_1} yields ${ED(D-1)}$ constraints, which grows linearly in $E$ and quadratically in $D$, and the condition in~\eqref{eq:indexing_cons_2} yields $K$ constraints, which grows linearly in $K$. 
Each of the conditions~\eqref{eq:T_non_neg}--\eqref{eq:L_geq_1} contributes as many constraints as the number of corresponding variables. Specifically,~\eqref{eq:T_non_neg} yields ${2^K-1}$ constraints,~\eqref{eq:I_non_neg} yields ${E(2^{K-D}-1)(3^D-2^D)}$ constraints,~\eqref{eq:J_non_neg} yields ${ED((D+1)2^{D-2}-D)}$ constraints, and~\eqref{eq:L_geq_1} yields one constraint.

We now specialize to the full demand family, ${E=\binom{K}{D}}$. 
Without exploiting automorphisms, the number of variables $T_{\mathrm{U}}$ grows exponentially in $K$, while the numbers of variables $I^{[\mathrm{W}_{j}]}_{\mathrm{U},\mathrm{V}}$ and $J^{[\mathrm{W}_{j}]}_{\mathrm{V},i}(k)$ grow linearly with $\binom{K}{D}$ and exponentially in $D$.
Similarly, the number of constraints scales linearly with $\binom{K}{D}$, with dominant terms exponential in $K$ and $D$. 

We next exploit the automorphisms of the full demand family.
Here, the automorphisms are exactly the relabelings of the message indices.
Since the full demand family contains all $D$-subsets of the $K$ messages, such relabelings leave the formulation unchanged.
They induce equivalence classes, or orbits, on the subset configurations that index the variables.
Rather than treating all equivalent configurations separately, we keep one representative from each orbit, as described next.

By symmetry, $T_{\mathrm{U}}$ depends only on ${u\coloneqq |\mathrm{U}|}$, and hence can be represented by variables $\widetilde{T}_u$, ${u\in[1:K]}$, yielding $K$ variables.
Similarly, $I^{[\mathrm{W}_{j}]}_{\mathrm{U},\mathrm{V}}$ depends only on ${u_1\coloneqq |\mathrm{U}\setminus \mathrm{W}_j|}$, ${u_2\coloneqq |\mathrm{U}\cap \mathrm{W}_j|}$, and ${v\coloneqq |\mathrm{V}|}$, and can be represented by variables $\widetilde{I}_{\{u_1,u_2\},v}$. 
The number of such variables is ${(K-D)\binom{D+1}{2}}$, which grows linearly in $K$ and quadratically in $D$.
Moreover, the variables ${J^{[\mathrm{W}_{j}]}_{\mathrm{V},i}(k)}$ depend only on ${v=|\mathrm{V}|}$ and $k$, and can be represented by variables $\widetilde{J}_v(k)$. 
The number of such variables is ${\binom{D}{2}}$, which grows quadratically in $D$.

\begin{figure}[!t]
\centering
\begin{minipage}{\columnwidth}
\begin{align}
& \sum_{v=1}^{D-u_2} \binom{D-u_2}{v}
\widetilde{I}_{\{u_1,u_2\},v}
\nonumber\\
& \quad
+
(N-1)
\sum_{v=1}^{u_2} \binom{u_2}{v}
\widetilde{I}_{\{u_1,u_2-v\},v}
\leq \widetilde{T}_{u_1+u_2},
\nonumber\\
& \hspace{1.5em}
\forall u_1\in[1:K-D],\
\forall u_2\in[0:D],
\label{eq:set_1_cons_reduced}
\\[0.4em]
& \widetilde{T}_{1}
+
(N-1)
\sum_{u_1=1}^{K-D}
\sum_{u_2=0}^{D-1}
\binom{K-D}{u_1}
\binom{D-1}{u_2}
\widetilde{I}_{\{u_1,u_2\},1}
\nonumber\\
& \quad
+
\sum_{v=2}^{D}
\sum_{k=v}^{D}
\binom{D-1}{v-1}
\widetilde{J}_{v}(k)
=
\frac{L}{N},
\label{eq:set_2_cons_reduced}
\\[0.4em]
& \widetilde{T}_{v}
+
(N-1)
\sum_{u_1=1}^{K-D}
\sum_{u_2=0}^{D-v}
\binom{K-D}{u_1}
\binom{D-v}{u_2}
\widetilde{I}_{\{u_1,u_2\},v}
\nonumber\\
& \quad
\geq
v
\sum_{k=v}^{D}
\widetilde{J}_{v}(k),
\quad
\forall v\in[2:D],
\label{eq:set_3_cons_reduced}
\\[0.4em]
& (N-1)\widetilde{T}_{1}
+
(N-1)^2
\sum_{k=1}^{m}
\sum_{u_1=1}^{K-D}
\binom{K-D}{u_1}
\binom{D-1}{k-1}
\nonumber\\
& \quad
\times
\widetilde{I}_{\{u_1,k-1\},1}
+
(N-1)
\sum_{k=2}^{m}
\sum_{v=1}^{k-1}
\binom{D-1}{v}
\widetilde{J}_{v+1}(k)
\nonumber\\
& \geq
N
\sum_{k=2}^{m+1}
\sum_{v=1}^{k-1}
\sum_{u_1=1}^{K-D}
\binom{K-D}{u_1}
\binom{D-1}{v}
\binom{D-v-1}{k-v-1}
\nonumber\\
& \quad
\times
\widetilde{I}_{\{u_1,k-v\},v}
+
\sum_{k=2}^{m+1}
\sum_{v=1}^{k-1}
v\binom{D-1}{v}
\widetilde{J}_{v+1}(k),
\nonumber\\
& \hspace{1.5em}
\forall m\in[1:D-1],
\label{eq:indexing_cons_1_reduced}
\\[0.4em]
& \sum_{u=1}^{K}
\binom{K-1}{u-1}
\widetilde{T}_{u}
\leq L,
\label{eq:indexing_cons_2_reduced}
\\[0.4em]
& \widetilde{T}_{u}\in\mathbb{N}_0,
\quad
\forall u\in[1:K],
\label{eq:T_non_neg_reduced}
\\[0.4em]
& \widetilde{I}_{\{u_1,u_2\},v}
\in\mathbb{N}_0,
\quad
\forall v\in[1:D],\
\forall u_1\in[1:K-D],
\nonumber\\
& \hspace{1.5em}
\forall u_2\in[0:D-v],
\label{eq:I_non_neg_reduced}
\\[0.4em]
& \widetilde{J}_{v}(k)\in\mathbb{N}_0,
\quad
\forall v\in[2:D],\
\forall k\in[v:D],
\label{eq:J_non_neg_reduced}
\\[0.4em]
& L\in\mathbb{N}.
\label{eq:L_geq_1_reduced}
\end{align}
\end{minipage}
\vspace{-0.4cm}
\end{figure}

Under this reduction, the optimization problem can be expressed in terms of the variables $\widetilde{T}_u$, $\widetilde{I}_{\{u_1,u_2\},v}$, and $\widetilde{J}_v(k)$, subject to the constraints~\eqref{eq:set_1_cons_reduced}--\eqref{eq:L_geq_1_reduced}, which are obtained by rewriting the constraints~\eqref{eq:set_1_cons}--\eqref{eq:L_geq_1} in terms of the reduced variables. 
We note that the symmetry leads to a significant reduction in the number of constraints. 
Specifically, condition~\eqref{eq:set_1_cons_reduced} yields ${(K-D)(D+1)}$ constraints, which grows linearly in $K$ and $D$. Condition~\eqref{eq:set_2_cons_reduced} yields one constraint. 
Conditions~\eqref{eq:set_3_cons_reduced} and~\eqref{eq:indexing_cons_1_reduced} yield ${D-1}$ constraints each, and~\eqref{eq:indexing_cons_2_reduced} yields one constraint. 
Conditions~\eqref{eq:T_non_neg_reduced}--\eqref{eq:L_geq_1_reduced} contribute $K$, ${(K-D)\binom{D+1}{2}}$, $\binom{D}{2}$, and one constraint, respectively. 

The above counts show that, after exploiting automorphisms, the numbers of variables and constraints scale polynomially in $K$ and $D$, specifically, linearly in $K$ and quadratically in $D$, rather than exponentially as in the original formulation.

We note that this reduced formulation subsumes the optimization framework introduced in~\cite{HWS2025}. 
Specifically, the framework in~\cite{HWS2025} is recovered by imposing the following additional restrictions: 
\begin{itemize}
\item For all ${u_1\in[1:K-D]}$, ${u_2\geq 1}$, and ${v\in[1:D]}$, 
\[{\widetilde{I}_{\{u_1,u_2\},v}=0},\]  so every side symbol contains only interference messages.
\item For all ${u_1\in[1:K-D]}$ and ${v\in[1:D]}$, 
\[
\widetilde{I}_{(u_1,0),v} =\frac{\widetilde{T}_{u_1+v}}{N-1}.
\]
This means that, to cancel the interference in target symbols with a fixed support at one server, the scheme uses, as side symbols, all symbols from the other ${N-1}$ servers that contain exactly the corresponding ${u_1}$ interference messages.
\item For all ${v\in[2:D]}$ and ${k>v}$, 
\[{\widetilde{J}_v(k)=0},\]
so each demand-only symbol with support size ${v}$ is recovered only in round ${v}$. 
\item For all ${v\in[2:D]}$, 
\[{\widetilde{J}_v(v)=\frac{1}{v}\sum_{i=0}^{K-D}\binom{K-D}{i}\widetilde{T}_{v+i}}.\]
This means that, after interference cancellation, the demand-only symbols recover equal numbers of subpackets from the demand messages involved in those symbols.
\end{itemize}

Although the formulation presented here admits a larger feasible class, it yields the same achievable rate as~\cite{HWS2025} in all tested full-demand-family instances. 
For general demand families, however, these additional degrees of freedom can be beneficial, as shown by the illustrative example in Section~VI.

\section{Open Problems and Future Directions}

Several important questions remain open regarding the fundamental limits and design of PSSR schemes.

The tightness of the converse bounds and the optimality of the achievable schemes are not known in general, either in terms of retrieval rate or subpacketization level.

Another open problem is to develop sharper converse bounds for structured classes of schemes. 
The rate converse derived here applies to arbitrary PSSR schemes and therefore does not exploit the balanced ${\{0,1\}}$-linear structure of the proposed construction, or linearity more broadly.

It also remains to identify other structured demand families, beyond the full and contiguous block demand families, for which the achievable scheme admits a closed-form characterization or can be derived from a substantially reduced optimization formulation.

Finally, the optimal tradeoff between retrieval rate and subpacketization level remains open. 
For a given lower bound on the retrieval rate, one may ask for the least required subpacketization level. 
Conversely, for a given upper bound on the subpacketization level, one may ask for the largest achievable retrieval rate.

Beyond these open problems, several broader directions remain for future work.

First, the PSSR setting considered in this work is prior-agnostic and assumes equal-length messages. 
Consequently, the converse and achievability results hold for any full-support prior on the demand family, but do not address unequal message lengths. 
This limitation is particularly important in view of semantic PIR with unequal message sizes, where non-uniform demand priors can affect the optimal retrieval rate~\cite{VBU2022}. 
A natural future direction is to study prior-aware PSSR with unequal message lengths, incorporating both the demand distribution and message heterogeneity into the converse bound and achievable scheme.

Another direction is to move beyond fixed-size demand sets. 
For example, in a medical dataset, a research company may wish to retrieve the records of all patients diagnosed with a given disease in order to perform an inference task. 
Since different diseases may be associated with different numbers of patients, the corresponding demand sets need not have the same cardinality. 
This motivates extending PSSR to such variable-size demand families.

The PSSR setting can also be generalized by allowing the user to possess side information.
In many practical settings, the user may already have access to some messages, or to functions of them, through prior interactions with the servers or with other users. 
In classical PIR and MPIR, such side information is known to yield more efficient schemes under various privacy requirements~\cite{KGHERS2017,KGHERS2017No0,HKGRS2018,SSM2018,LG2018,HKS2018,HKS2019Journal,HKS2019,KKHS12019,KKHS22019,CWJ2020,LG2020CISS,KH2021Journal,LJ2022,WHS2024,EH2024}. 
This motivates investigating whether side information can similarly improve efficiency in PSSR, and extending the optimization framework to accommodate such settings.

Furthermore, in related settings, protecting each demand message separately, rather than protecting the demand set as a whole, is known to yield more efficient schemes~\cite{HKRS2019,HS2021,HS2022LinCap}. 
This suggests studying whether analogous relaxations of the privacy requirement can also improve the efficiency of PSSR schemes. 
For certain structured demand families, such relaxations may yield gains beyond those previously shown for the full demand family, leading to substantially higher retrieval rates or much lower subpacketization levels.

Finally, the optimization framework developed here is restricted to a combinatorial class of linear combinations and therefore does not exploit the full algebraic design space of coded combinations. 
Since more algebraic constructions, such as those based on MDS-coded combinations, are known to provide gains in related settings~\cite{BU2018}, it would be interesting to extend the framework to more general coded combinations and investigate whether they can improve upon the current construction. 
This may reveal new rate--subpacketization tradeoffs beyond the current ${\{0,1\}}$-linear formulation.

\balance

\bibliographystyle{IEEEtran}
\bibliography{PIR_PC_Refs}

\end{document}